\documentclass[journal,draftcls,onecolumn]{IEEEtran}
\usepackage[T1]{fontenc}

\ifCLASSINFOpdf

\else

\fi
\usepackage{amsmath}
\usepackage{ifthen}
\usepackage{tikz}
\usepackage{cite}
\usepackage{amsthm,amsfonts,amssymb}
\usepackage{graphicx}
\usepackage{subfig}
\usepackage{blkarray}
\usepackage{dsfont}
\usepackage[mathscr]{euscript}
\usepackage{enumitem}
\usepackage{algorithm}
\usepackage[noend]{algpseudocode}
\usepackage{blkarray}
\usepackage{stfloats}%
\newtheorem{theorem}{Theorem}[section]

\newtheorem{proposition}{Proposition}[section]
\newtheorem{lemma}{Lemma}[section]
\theoremstyle{remark}
\newtheorem*{remark}{Remark}
\theoremstyle{definition}
\newtheorem{definition}{Definition}[section]

\usepackage{url}
\usepackage{arydshln}
\usepackage{mathtools}
\usepackage{dsfont}

\definecolor{brickred}{cmyk}{0,0.89,0.94,0.28}
\definecolor{goldenrod}{cmyk}{0,0.10,0.84,0}
\definecolor{purple}{cmyk}{0.45,0.86,0,0}
\definecolor{rawsienna}{cmyk}{0,0.72,1,0.45}
\definecolor{olivegreen}{cmyk}{0.64,0,0.95,0.40}
\definecolor{peach}{cmyk}{0,0.5,0.7,0}
\definecolor{darkolive}{rgb}{0.,0.4,0.}
\colorlet{grey}{gray!40}

\usepackage[cmintegrals]{newtxmath}

\begin{document}

\title{Coding Schemes Based on Reed-Muller Codes for $(d,\infty)$-RLL Input-Constrained Channels}
%
%
%

\author{V.~Arvind~Rameshwar,~\IEEEmembership{Student Member,~IEEE,}
        and~Navin~Kashyap,~\IEEEmembership{Senior~Member,~IEEE}
\thanks{The work of V.~A.~Rameshwar was supported by a Prime Minister's Research Fellowship, from the Ministry of Education, Govt. of India. {Parts of this work were presented at the 2022 IEEE International Symposium on Information Theory (ISIT) and the 2022 IEEE Information Theory Workshop (ITW)}.}
\thanks{The authors are with the Department of Electrical Communication Engineering, Indian Institute of Science, Bengaluru 560012, India (e-mail: vrameshwar@iisc.ac.in;~nkashyap@iisc.ac.in).}
}

\markboth{IEEE Transactions on Information Theory,~Vol.XX, No.~XX, Month and Year}%
{Rameshwar and Kashyap: Reed-Muller Codes for $(d,\infty)$-RLL Input-Constrained Channels}

\maketitle

\begin{abstract}
The paper considers coding schemes derived from Reed-Muller (RM) codes, for transmission over input-constrained memoryless channels. Our focus is on the $(d,\infty)$-runlength limited (RLL) constraint, which mandates that any pair of successive $1$s be separated by at least $d$ $0$s. 
In our study, we first consider $(d,\infty)$-RLL subcodes of RM codes, taking the coordinates of the RM codes to be in the standard lexicographic ordering. We show, via a simple construction, that RM codes of rate $R$ have linear $(d,\infty)$-RLL subcodes of rate $R\cdot{2^{-\left \lceil \log_2(d+1)\right \rceil}}$. We then show that our construction is essentially rate-optimal, by deriving an upper bound on the rates of linear $(d,\infty)$-RLL subcodes of RM codes of rate $R$. Next, for the special case when $d=1$, we prove the existence of potentially non-linear $(1,\infty)$-RLL subcodes that achieve a rate of $\max\left(0,R-\frac38\right)$. This, for $R > 3/4$, beats the $R/2$ rate obtainable from linear subcodes. We further derive upper bounds on the rates of $(1,\infty)$-RLL subcodes, not necessarily linear, of a certain canonical sequence of RM codes of rate $R$. We then shift our attention to settings where the coordinates of the RM code are not ordered according to the lexicographic ordering, and derive rate upper bounds for linear $(d,\infty)$-RLL subcodes in these cases as well. 
Finally, we present a new two-stage constrained coding scheme, again using RM codes of rate $R$, which outperforms any linear coding scheme using $(d,\infty)$-RLL subcodes, for values of $R$ close to $1$. 
\end{abstract}

\begin{IEEEkeywords}
Reed-Muller codes, runlength-limited constraints, BMS channels
\end{IEEEkeywords}

\IEEEpeerreviewmaketitle

\section{Introduction}
\label{sec:intro}
%
%
%
%
\IEEEPARstart{C}{onstrained} coding is a method that helps alleviate error-prone sequences in data recording and communication systems, by allowing the encoding of arbitrary user data sequences into only those sequences that respect the constraint (see, for example, \cite{Roth} or \cite{Immink}). In this work, we consider the construction of constrained codes derived from the Reed-Muller (RM) family of codes.

Our focus is on the $(d,\infty)$-runlength limited (RLL) input constraint, which mandates that there be at least $d$ $0$s between every pair of successive $1$s in the binary input sequence. For example, when $d=2$, the sequence $100100010$ respects the $(2,\infty)$-RLL constraint, but the sequence $10100010$, does not. It can also be checked that the $(1,\infty)$-RLL constraint is the same as a ``no-consecutive-ones'' constraint. The $(d,\infty)$-RLL constraint is a special case of the $(d,k)$-RLL constraint, which admits only binary sequences in which successive $1$s are separated by at least $d$ $0$s, and the length of any run of $0$s is at most $k$. In a magnetic recording system, for example, the $(d,\infty)$-RLL constraint on the data sequence (where the bit $1$ corresponds to a voltage peak of high amplitude and the bit $0$ corresponds to no peak) ensures that successive $1$s are spaced far enough apart, so that there is little inter-symbol interference between the voltage responses corresponding to the magnetic transitions. Reference \cite{Immink2} contains many examples of $(d,k)$-RLL codes used in practice in magnetic storage and recording. More recently, $(d,k)$-RLL input constrained sequences have also been investigated for joint energy and information transfer performance \cite{infoenergy}.

We are interested in the transmission of $(d,\infty)$-RLL constrained codes over an input-constrained, noisy binary-input memoryless symmetric (BMS) channel, without feedback. BMS channels form a subclass of discrete memoryless channels (DMCs). Examples of BMS channels include the binary erasure channel (BEC) and binary symmetric channel (BSC), shown in Figures \ref{fig:bec} and \ref{fig:bsc}. Figure \ref{fig:gen_const_DMC} illustrates our system model of an input-constrained BMS channel, without feedback. Input-constrained DMCs fall under the broad class of discrete finite-state channels (DFSCs). 

\begin{figure}[!h]
	\centering
	\subfloat[]{
		\resizebox{0.21\textwidth}{!}{

			\tikzset{every picture/.style={line width=0.75pt}} 
			
			\begin{tikzpicture}[x=0.75pt,y=0.75pt,yscale=-1,xscale=1]
				
				\draw    (454.24,136.67) -- (521.85,136.67) ;
				\draw [shift={(523.85,136.67)}, rotate = 180] [color={rgb, 255:red, 0; green, 0; blue, 0 }  ][line width=0.75]    (10.93,-3.29) .. controls (6.95,-1.4) and (3.31,-0.3) .. (0,0) .. controls (3.31,0.3) and (6.95,1.4) .. (10.93,3.29)   ;
				\draw    (454.24,136.67) -- (519.94,155.67) ;
				\draw [shift={(521.86,156.22)}, rotate = 196.13] [color={rgb, 255:red, 0; green, 0; blue, 0 }  ][line width=0.75]    (10.93,-3.29) .. controls (6.95,-1.4) and (3.31,-0.3) .. (0,0) .. controls (3.31,0.3) and (6.95,1.4) .. (10.93,3.29)   ;
				\draw    (454.24,173.79) -- (521.85,173.79) ;
				\draw [shift={(523.85,173.79)}, rotate = 180] [color={rgb, 255:red, 0; green, 0; blue, 0 }  ][line width=0.75]    (10.93,-3.29) .. controls (6.95,-1.4) and (3.31,-0.3) .. (0,0) .. controls (3.31,0.3) and (6.95,1.4) .. (10.93,3.29)   ;
				\draw    (454.24,173.79) -- (519.93,156.72) ;
				\draw [shift={(521.86,156.22)}, rotate = 165.43] [color={rgb, 255:red, 0; green, 0; blue, 0 }  ][line width=0.75]    (10.93,-3.29) .. controls (6.95,-1.4) and (3.31,-0.3) .. (0,0) .. controls (3.31,0.3) and (6.95,1.4) .. (10.93,3.29)   ;
				
				\draw (476.48,146) node [anchor=north west][inner sep=0.75pt]  [font=\footnotesize]  {$\epsilon $};
				\draw (476.48,158) node [anchor=north west][inner sep=0.75pt]  [font=\footnotesize]  {$\epsilon $};
				\draw (477.13,123.4) node [anchor=north west][inner sep=0.75pt]  [font=\footnotesize]  {$1-\epsilon $};
				\draw (477.13,175) node [anchor=north west][inner sep=0.75pt]  [font=\footnotesize]  {$1-\epsilon $};
				\draw (443.86,131.32) node [anchor=north west][inner sep=0.75pt]  [font=\footnotesize]  {${\displaystyle 0}$};
				\draw (443.86,166.96) node [anchor=north west][inner sep=0.75pt]  [font=\footnotesize]  {$1$};
				\draw (525.01,131.32) node [anchor=north west][inner sep=0.75pt]  [font=\footnotesize]  {$1$};
				\draw (525.01,149.88) node [anchor=north west][inner sep=0.75pt]  [font=\footnotesize]  {$0$};
				\draw (525.01,167.7) node [anchor=north west][inner sep=0.75pt]  [font=\footnotesize]  {$-1$};

			\end{tikzpicture}
		}
		\label{fig:bec}
	}
	\qquad
	\subfloat[]{
		\resizebox{0.2\textwidth}{!}{

			\tikzset{every picture/.style={line width=0.75pt}} 
			
			\begin{tikzpicture}[x=0.75pt,y=0.75pt,yscale=-1,xscale=1]
				
				\draw    (449.24,119.67) -- (516.85,119.67) ;
				\draw [shift={(518.85,119.67)}, rotate = 180] [color={rgb, 255:red, 0; green, 0; blue, 0 }  ][line width=0.75]    (10.93,-3.29) .. controls (6.95,-1.4) and (3.31,-0.3) .. (0,0) .. controls (3.31,0.3) and (6.95,1.4) .. (10.93,3.29)   ;
				\draw    (449.24,119.67) -- (517.09,155.85) ;
				\draw [shift={(518.85,156.79)}, rotate = 208.07] [color={rgb, 255:red, 0; green, 0; blue, 0 }  ][line width=0.75]    (10.93,-3.29) .. controls (6.95,-1.4) and (3.31,-0.3) .. (0,0) .. controls (3.31,0.3) and (6.95,1.4) .. (10.93,3.29)   ;
				\draw    (449.24,156.79) -- (516.85,156.79) ;
				\draw [shift={(518.85,156.79)}, rotate = 180] [color={rgb, 255:red, 0; green, 0; blue, 0 }  ][line width=0.75]    (10.93,-3.29) .. controls (6.95,-1.4) and (3.31,-0.3) .. (0,0) .. controls (3.31,0.3) and (6.95,1.4) .. (10.93,3.29)   ;
				\draw    (449.24,156.79) -- (517.09,120.61) ;
				\draw [shift={(518.85,119.67)}, rotate = 151.93] [color={rgb, 255:red, 0; green, 0; blue, 0 }  ][line width=0.75]    (10.93,-3.29) .. controls (6.95,-1.4) and (3.31,-0.3) .. (0,0) .. controls (3.31,0.3) and (6.95,1.4) .. (10.93,3.29)   ;
				
				\draw (485.05,124) node [anchor=north west][inner sep=0.75pt]  [font=\footnotesize]  {$p$};
				\draw (483.48,144) node [anchor=north west][inner sep=0.75pt]  [font=\footnotesize]  {$p$};
				\draw (471.13,104.4) node [anchor=north west][inner sep=0.75pt]  [font=\footnotesize]  {$1-p$};
				\draw (470.13,159.4) node [anchor=north west][inner sep=0.75pt]  [font=\footnotesize]  {$1-p$};
				\draw (438.86,114.32) node [anchor=north west][inner sep=0.75pt]  [font=\footnotesize]  {${\displaystyle 0}$};
				\draw (438.86,149.96) node [anchor=north west][inner sep=0.75pt]  [font=\footnotesize]  {$1$};
				\draw (520.01,114.32) node [anchor=north west][inner sep=0.75pt]  [font=\footnotesize]  {$1$};
				\draw (520.01,150.7) node [anchor=north west][inner sep=0.75pt]  [font=\footnotesize]  {$-1$};

			\end{tikzpicture}
		}
		\label{fig:bsc}}
	\caption{(a) The binary erasure channel (BEC$(\epsilon)$) with erasure probability $\epsilon$ and output alphabet $\mathscr{Y} = \{-1,0,1\}$. (b) The binary symmetric channel (BSC$(p)$) with crossover probability $p$ and output alphabet $\mathscr{Y} = \{-1,1\}$.}
\end{figure}
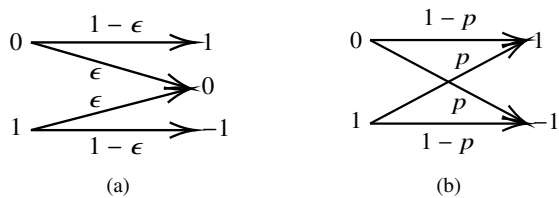

For the setting of \emph{unconstrained} DMCs without feedback, explicit codes achieving the capacities, or whose rates achieved are very close to the capacities, have been derived in works such as \cite{polar, kud1, luby, ru1, kud2}. However, to the best of our knowledge, there are no explicit coding schemes with provable rate and probability of error guarantees, for input-constrained DMCs without feedback. We mention also that unlike the case of the unconstrained DMC, whose capacity is characterized by Shannon's  single-letter formula, $C_{\text{DMC}} = \sup_{P(x)} I(X;Y)$, the computation of the capacity of a DFSC, given by the maximum mutual information rate between inputs and outputs, is a much more difficult problem to tackle, and is open for many simple channels. In fact, the computation of the mutual information rate even for the simple case of Markov inputs, over input-constrained DMCs, reduces to the computation of the entropy rate of a hidden Markov process---a well-known hard problem (see \cite{ZW88,Hanerasure,Han1,erikordentlich,Ordentlich,arnk21ncc, arnk20}).

In this paper, we make progress on the construction of explicit codes over $(d,\infty)$-RLL input-constrained DMCs. Our motivation for constructing $(d,\infty)$-RLL constrained codes using RM codes, are the very recent results of Reeves and Pfister \cite{Reeves} and Abbe and Sandon \cite{abbesandon} that Reed-Muller (RM) codes achieve, under bitwise maximum a-posteriori probability (bit-MAP) {and blockwise maximum a-posteriori probability (block-MAP) decoding, respectively,} any rate $R\in [0,C)$, over unconstrained BMS channels, where $C$ is the capacity of the BMS channel.
We note that for the specific setting of the BEC, Kudekar et al. in \cite{kud1} were the first to show that Reed-Muller (RM) codes are capacity-achieving. As a consequence of these results, the constrained coding schemes we construct have {bit- and block-error probabilities going to $0$ (using the bit-MAP and block-MAP decoders for the RM code, respectively)}, as the blocklength goes to infinity, if the codes are used over $(d,\infty)$-RLL input-constrained BMS channels.
\subsection{Other Approaches From Prior Art}

There is extensive literature on constrained code constructions that can correct a fixed number of errors, under a combinatorial error model (see Chapter 9 in \cite{Roth} and the references therein). Another related line of work that aims to limit error propagation during the decoding of constrained codes can be found in \cite{Bliss, Mansuripur, Imm97, FC98}. More recently, the work \cite{deSouza} proposed the insertion of parity bits for error correction, into those positions of binary constrained codewords that are unconstrained, i.e., whose bits can be flipped, with the resultant codeword still respecting the constraint. However, such works do not analyze the resilience of constrained codes to stochastic channel noise---the standard noise model in information theory, which is taken up in our paper. 

To the best of our knowledge, the paper \cite{pvk}, on rates achievable by $(d,k)$-RLL subcodes of cosets of a linear block code, was the first to consider the problem of coding schemes over RLL input-constrained DMCs. Specifically, from Corollary 1 of \cite{pvk}, we see that there exist cosets of capacity-achieving (over the unconstrained BMS channel) codes, whose $(d,k)$-RLL constrained subcodes have rate at least ${\kappa_{d,k}} + C -1$, where ${\kappa_{d,k}}$ is the noiseless capacity of the $(d,k)$-RLL constraint. However, the work in \cite{pvk} does not identify \emph{explicit} codes over RLL input-constrained channels. 

In the context of designing coding schemes over (input-constrained) BMS channels, it would be remiss to not comment on the rates achieved by polar codes---another family of explicit codes that is capacity-achieving over a broad class of channels (see, for example, \cite{polar,tal2, sasoglutal1}, and references therein). Following the work of Li and Tan in \cite{litan}, we have that the capacity without feedback of the class of input-constrained DMCs can be approached arbitrarily closely using stationary, ergodic Markov input distributions of finite order. Moreover, from the results in \cite{sasoglutal1}, it follows that polar codes can achieve the mutual information rate of any stationary, ergodic finite-state Markov input process, over any DMC. In particular, this shows that polar codes achieve the capacity of $(d,\infty)$-RLL input-constrained DMCs, which includes the class of $(d,\infty)$-RLL input-constrained BMS channels. However, this observation is not very helpful for the following reasons:
\begin{itemize}
	\item We do not possess knowledge of an optimal sequence of Markov input distributions.
	\item The polar code construction described above is not explicit, since the choice of bit-channels to send information bits over is not explicitly known for an arbitrary BMS channel.
	\item It is hard to compute the rate achieved by such a coding scheme, since such a computation reduces to the calculation of the mutual information rate of a hidden Markov process.
\end{itemize}



\begin{figure}[!t]
	\centering
	\resizebox{0.7\textwidth}{!}{

		\tikzset{every picture/.style={line width=0.75pt}} 
		
		\begin{tikzpicture}[x=0.75pt,y=0.75pt,yscale=-1,xscale=1]
			
			\draw   (321,106.32) -- (391,106.32) -- (391,174.32) -- (321,174.32) -- cycle ;
			\draw    (390.5,137.65) -- (438.5,137.65) ;
			\draw [shift={(440.5,137.65)}, rotate = 180] [color={rgb, 255:red, 0; green, 0; blue, 0 }  ][line width=0.75]    (10.93,-3.29) .. controls (6.95,-1.4) and (3.31,-0.3) .. (0,0) .. controls (3.31,0.3) and (6.95,1.4) .. (10.93,3.29)   ;
			\draw   (441,105) -- (511,105) -- (511,180) -- (441,180) -- cycle ;
			\draw   (561,117.65) -- (631,117.65) -- (631,157.65) -- (561,157.65) -- cycle ;
			\draw    (510.5,137.65) -- (558.5,137.65) ;
			\draw [shift={(560.5,137.65)}, rotate = 180] [color={rgb, 255:red, 0; green, 0; blue, 0 }  ][line width=0.75]    (10.93,-3.29) .. controls (6.95,-1.4) and (3.31,-0.3) .. (0,0) .. controls (3.31,0.3) and (6.95,1.4) .. (10.93,3.29)   ;
			\draw    (630.5,137.65) -- (678.5,137.65) ;
			\draw [shift={(680.5,137.65)}, rotate = 180] [color={rgb, 255:red, 0; green, 0; blue, 0 }  ][line width=0.75]    (10.93,-3.29) .. controls (6.95,-1.4) and (3.31,-0.3) .. (0,0) .. controls (3.31,0.3) and (6.95,1.4) .. (10.93,3.29)   ;
			\draw    (250.5,136.65) -- (318.5,136.65) ;
			\draw [shift={(320.5,136.65)}, rotate = 180] [color={rgb, 255:red, 0; green, 0; blue, 0 }  ][line width=0.75]    (10.93,-3.29) .. controls (6.95,-1.4) and (3.31,-0.3) .. (0,0) .. controls (3.31,0.3) and (6.95,1.4) .. (10.93,3.29)   ;

			\draw (461,151.01) node [anchor=north west][inner sep=0.75pt]  [font=\normalsize]  {$P_{Y|X}$};
			\draw (649,121.01) node [anchor=north west][inner sep=0.75pt]  [font=\normalsize]  {$\hat{m}$};
			\draw (529,120.01) node [anchor=north west][inner sep=0.75pt]  [font=\normalsize]  {$y^{n}$};
			\draw (409,121.01) node [anchor=north west][inner sep=0.75pt]  [font=\normalsize]  {$x^{n}$};
			\draw (571,132.61) node [anchor=north west][inner sep=0.75pt]  [font=\normalsize] [align=left] {Decoder};
			\draw (323,124) node [anchor=north west][inner sep=0.75pt]   [align=left] {{\normalsize Constrained}\\{\normalsize \ \ Encoder}};
			\draw (461,123) node [anchor=north west][inner sep=0.75pt]   [align=left] {BMS};
			\draw (255,116) node [anchor=north west][inner sep=0.75pt]  [font=\normalsize]  {${\displaystyle m\in \left[ 2^{nR}\right]}$};

		\end{tikzpicture}
	}	
	\caption{System model of an input-constrained binary-input memoryless symmetric (BMS) channel without feedback.}
	\label{fig:gen_const_DMC}
\end{figure}
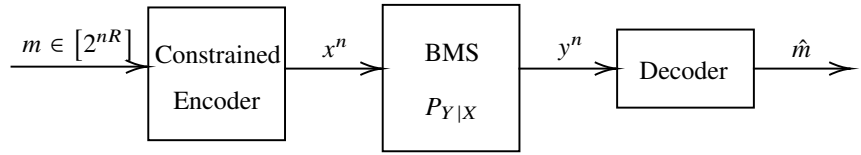

\subsection{Our Contributions}
In this work, we propose explicit coding schemes using RM codes, with computable rate lower bounds, over input-constrained BMS channels. In the first part of the paper, we fix the ordering of the coordinates of the RM codes we consider to be the standard lexicographic ordering. Suppose that $C$ is the capacity of the unconstrained BMS channel. Our first approach to designing $(d,\infty)$-RLL constrained codes is simply to identify \emph{linear} $(d,\infty)$-RLL subcodes of RM codes of rate $R$. Our results identify subcodes of rate $R\cdot{2^{-\left \lceil \log_2(d+1)\right \rceil}}$, in the limit as the blocklength goes to infinity, thereby showing that a rate of $C\cdot{2^{-\left \lceil \log_2(d+1)\right \rceil}}$ is achievable over $(d,\infty)$-RLL input-constrained BMS channels, using linear RM subcodes. Next, we present a lower bound on rates of \emph{non-linear} $(1,\infty)$-RLL subcodes of RM codes. Our result shows that a rate of $\max\left(0,C-\frac{3}{8}\right)$ is achievable, using this approach. {Note that the rates mentioned are achievable under both bit-MAP (from \cite{Reeves}) and block-MAP decoding (from \cite{abbesandon}).}

We then consider the question of obtaining upper bounds on the rates of $(d,\infty)$-RLL subcodes of RM codes of rate $R$, under the lexicographic coordinate ordering. First, we impose the additional restriction that the subcodes be linear. We prove that any linear $(d,\infty)$-RLL subcode can have rate at most $\frac{R}{d+1}$, in the limit as the blocklength goes to infinity. This shows that the linear $(d,\infty)$-RLL subcodes we previously identified are essentially rate-optimal. Next, we consider general (not necessarily linear) constrained subcodes of RM codes, and derive an upper bound on the rate of the largest $(1,\infty)$-RLL subcodes of a certain canonical sequence of RM codes of rate $R$, which we had also used in our lower bounds. Our novel method of analysis involves using an alternative characterization of $(1,\infty)$-RLL codewords of RM codes, and employs properties of the weight distribution of RM codes---a topic that has received revived attention over the last decade (see, for example, the survey \cite{rm_survey} and the papers \cite{abbe1, kaufman, sberlo, anuprao}). Unfortunately, this method does not readily extend to the $(d,\infty)$-RLL case, upper bounds for which remain an open problem.

Since permutations of coordinates have the potential to convert a binary word that does not respect the $(d,\infty)$-RLL constraint to one that does, we then ask the question if under alternative coordinate orderings, we can obtain \emph{linear} $(d,\infty)$-RLL subcodes of RM codes of rate $R$, of rate larger than the upper bound of $\frac{R}{d+1}$ that we had derived for the case when the coordinates follow the lexicographic ordering. {With this in mind, we consider RM codes of large blocklength, whose coordinates have been ordered according to an arbitrary but fixed permutation. We show that for almost all coordinate permutations, linear $(d,\infty)$-RLL subcodes of such permuted large-blocklength RM codes must respect a rate upper bound of $\frac{R}{d+1}+\delta$, where $\delta>0$ can be arbitrarily small.}

As an improvement over the rates achievable using $(d,\infty)$-RLL {subcodes}, we propose a new explicit two-stage (or concatenated) coding scheme using RM codes. The rate achieved, {under block-MAP decoding}, by this scheme, is at least $\frac{{\kappa_{d}}\cdot C^2\cdot 2^{-\left \lceil \log_2(d+1)\right \rceil}}{C^2\cdot 2^{-\left \lceil \log_2(d+1)\right \rceil} + 1-C+\epsilon}$, where ${\kappa_{d}}$ is the noiseless capacity of the input constraint, and $\epsilon>0$ can be taken to be as small as needed. For example, when $d=1$, the rates achieved using this two-stage scheme are better than those achieved by any scheme that uses linear $(1,\infty)$-RLL subcodes of RM codes (under almost all coordinate orderings), when $C\gtrapprox 0.7613$, and better than the rate achieved by our non-linear subcodes for all $C\lessapprox 0.55$ and $C\gtrapprox 0.79$.  Moreover, as the capacity of the channel approaches $1$, i.e., as the channel noise approaches $0$, the rate achieved by our two-stage coding scheme can be made arbitrarily close to ${\kappa_{d}}$, which is the largest rate achievable, at zero noise, given the constraint.

The remainder of the paper is organized as follows: Section \ref{sec:notation} introduces the notation and provides some preliminary background. Section \ref{sec:main} states our main results. The achievable rates using $(d,\infty)$-RLL subcodes of RM codes of rate $R$, under the lexicographic coordinate ordering, are discussed in Section \ref{sec:rm}. Section \ref{sec:rmlinub} then describes the derivation of upper bounds on the rates of linear $(d,\infty)$-RLL subcodes of RM codes of rate $R$, and Section \ref{sec:rmub} contains a derivation of upper bounds on the rates of general $(1,\infty)$-RLL subcodes of a canonical sequence of RM codes of rate $R$, at all times assuming the ordering of coordinates to be the lexicographic ordering. Section \ref{sec:perm} then discusses rate upper bounds for linear $(d,\infty)$-RLL subcodes of RM codes of rate $R$, under other coordinate orderings. Section \ref{sec:cosets} then presents our construction of a two-stage or concatenated $(d,\infty)$-RLL constrained coding scheme using RM codes. Finally, Section \ref{sec:conclusion} contains concluding remarks and a discussion on possible future work.

\section{Notation and Preliminaries}
\label{sec:notation}
\subsection{Notation}
The notation $[n]$ denotes the set, $\{1,2,\ldots,n\}$, of integers, and the notation $[a:b]$, for $a<b$, denotes the set of integers $\{a,a+1,\ldots,b\}$. Moreover, for a real number $x$, we use $\left \lceil x \right \rceil$ to denote the smallest integer larger than or equal to $x$. For vectors $\mathbf{w}$ and $\mathbf{v}$ of length $n$ and $m$, respectively, we denote their concatenation by the $(m+n)$-length vector, $\mathbf{w}\mathbf{v}$. The notation 
$x^N$ denotes the vector $(x_1,\ldots,x_N)$. We also use the notation $\mathbf{e}_i^{(n)}$ to denote the standard basis vector of length $n$, with a $1$ at position $i$, and $0$s elsewhere, for $i\in [n]$. When $n=2^m$, for some $m$, we interchangeably index the coordinates of vectors $\mathbf{v}\in \{0,1\}^n$ by integers $i\in [0:n-1]$ and by $m$-tuples $\mathbf{b} = (b_1,\ldots,b_m)\in \{0,1\}^m$. We let {\textbf{B}$_m(i)$} denote the length-$m$ binary representation of $i$, for $0\leq i\leq 2^m-1$. {We drop the subscript `$m$' if the length of the representation is clear from the context.} We then use the notation $\mathbf{e}_\mathbf{b}^{(n)}$ to denote the standard basis vector of length $n$, with a $1$ at position $\mathbf{b}$, and $0$s elsewhere. The superscript `$(n)$' will be dropped when clear from {the} context. For any $n\in \mathbb{N}$, we denote by $S_{(d,\infty)}^{(n)}$, the set of all $n$-length binary words that respect the $(d,\infty)$-RLL constraint, and we set $S_{(d,\infty)}=\bigcup_{n\geq 1} S_{(d,\infty)}^{(n)}$. 

The notation $X\sim \mathcal{N}(\mu, \sigma^2)$ refers to the fact that the random variable $X$ is drawn according to the Gaussian distribution, with mean $\mu$ and variance $\sigma^2>0$. 
Also, the notation $h_b(p):=-p\log_2 p - (1-p)\log_2(1-p)$ is the binary entropy function, for $p\in [0,1]$. All through, the empty summation is defined to be $0$, and the empty product is defined to be $1$. 
We write exp$_2(z)$ for $2^z$, where $z\in \mathbb{R}$, and the notation $\ln$ refers to the natural logarithm. Throughout, we use the convenient notation $\binom{m}{\le r}$ to denote the summation $\sum\limits_{i=0}^r \binom{m}{i}$. For sequences $(a_n)_{n\geq 1}$ and $(b_n)_{n\geq 1}$ of positive reals, we say that $a_n = O(b_n)$, if there exists $n_0\in \mathbb{N}$ and a positive real $M$, such that $a_n\leq M\cdot b_n$, for all $n\geq n_0$. We say that $a_n = o(b_n)$, if $\lim_{n\to \infty} \frac{a_n}{b_n} = 0$.

\subsection{Information Theoretic Preliminaries}
\subsubsection{Block Codes and Constrained Codes}
We recall the following definitions of block codes and linear codes over $\mathbb{F}_2$ and their rates (see, for example, Chapter 1 of \cite{roth_coding_theory}).

\begin{definition}
	An $(n,M)$ block code $\mathcal{C}$ over $\mathbb{F}_2$ is a nonempty subset of $\mathbb{F}_2^n$, with $|\mathcal{C}| = M$. The rate of the block code $\mathcal{C}$ is given by
	\[
	\text{rate}(\mathcal{C}) := \frac{\log_2 M}{n}.
	\]
\end{definition}
Moreover, given a sequence of codes $\{\mathcal{C}^{(n)}\}_{n\geq 1}$, if rate$(\mathcal{C}^{(n)})\xrightarrow{n\to \infty} R$, for some $R\in [0,1]$, then we say that $\{\mathcal{C}^{(n)}\}_{n\geq 1}$ is of rate $R$.
{
\begin{definition}
	An $[n,k]$ linear code $\mathcal{C}$ over $\mathbb{F}_2$ is an $(n,2^k)$ block code that is a subspace of $\mathbb{F}_2^n$.
\end{definition}
}
Consider now the set $S_{(d,\infty)}^{(n)}$ of length-$n$ binary words that satisfy the $(d,\infty)$-RLL constraint. {Recall that this constraint admits only those binary sequences that have at least $d$ $0$s between successive $1$s.} We recall the following definition (see, for example, Chapter 3 of \cite{Roth}):
\begin{definition}
	The noiseless capacity ${\kappa_d}$ of the $(d,\infty)$-RLL constraint is defined as
	\[
	{\kappa_d}:= \lim_{n\to \infty} \frac{\log_2\left\lvert S_{(d,\infty)}^{(n)}\right \rvert}{n} = \inf_n \frac{\log_2\left\lvert S_{(d,\infty)}^{(n)}\right \rvert}{n},
	\]
where the last equality follows from the subadditivity of the sequence $\left(\log_2\left\lvert S_{(d,\infty)}^{(n)}\right\rvert\right)_{n\geq 1}$. 
\end{definition}
\subsubsection{Linear Codes Over BMS Channels}
\label{sec:lincodesBMS}

A binary-input channel $W$ is defined by the binary-input input alphabet $\mathscr{X} = \{0,1\}$, the output alphabet $\mathscr{Y}\subseteq {\mathbb{R}}$, and the transition probabilities $\left(P(y|x): x\in \mathscr{X}, y\in \mathscr{Y}\right)$, where $P(\cdot|x)$ is a density function with respect to the counting measure, if $\mathscr{Y}$ is discrete, and with respect to the Lebesgue measure, if $\mathscr{Y} = \mathbb{R}$ and the output distribution is continuous. A binary-input memoryless symmetric (BMS) channel obeys a memorylessness property, i.e., $P(y_i|x^{i},y^{i-1}) = P(y_i|x_i)$, for all $i$. Further, the channel is symmetric, in that $P(y|1) = P(-y|0)$, for all $y\in \mathscr{Y}$. Every such channel can be expressed as a multiplicative noise channel, in the following sense: if at any $i$ the input random symbol is $X_i\in \{0,1\}$, then the corresponding output symbol $Y_i\in \mathscr{Y}$ is given by
\[
Y_i = (-1)^{X_i}\cdot Z_i,
\]
where the noise random variables $Z^n$ are independent and identically distributed, and the noise process $(Z_i)_{i\geq 1}$ is independent of the input process $(X_i)_{i\geq 1}$. Common examples of such channels include the binary erasure channel (BEC$(\epsilon)$), with $P(Z_i = 1) = 1-\epsilon$ and $P(Z_i=0)=\epsilon$, the binary symmetric channel (BSC$(p)$), with $P(Z_i = 1) = 1-p$ and $P(Z_i=-1)=p$, and the binary additive white Gaussian noise (BI-AWGN) channel, where $Z_i\sim \mathcal{N}(1,\sigma^2)$. Figures \ref{fig:bec} and \ref{fig:bsc} depict the BEC and BSC, pictorially. 

In this work, we are interested in designing codes over input-constrained BMS channels. We refer the reader to \cite{Reeves} for definitions of the {bitwise maximum a-posteriori probability (bit-MAP)} decoder and bit-error probability $P_b^{(n)}$ (under bit-MAP decoding), indexed by the blocklength $n$ of the code. We also refer the reader to standard texts on information theory such as \cite{cover_thomas}, \cite{gallager} for definitions of blockwise maximum a-posteriori probability (block-MAP) decoding and block-error probability $P_B^{(n)}$.

Specifically, we shall be using constrained subcodes of linear codes $\left\{\mathcal{C}^{(n)}\right\}_{n\geq 1}$ over BMS channels. At the decoder end, we shall employ the bit-MAP {or block-MAP} decoders of $\left\{\mathcal{C}^{(n)}\right\}_{n\geq 1}$. Since $\mathcal{C}^{(n)}$ is linear, the average bit-error probability $P_b^{(n)}$ (resp. the average block-error probability $P_B^{(n)}$), under bit-MAP decoding (resp. block-MAP decoding) of $\mathcal{C}^{(n)}$ is the same as the average bit-error probability (resp. the average block-error probability) when a subcode of $\mathcal{C}^{(n)}$ is used (see \cite{Reeves} for a discussion). 

We say that a rate $R$ is achieved by $\left\{\mathcal{C}^{(n)}\right\}_{n\geq 1}$ over a BMS channel, under bit-MAP decoding {(resp. under block-MAP decoding)}, if $\text{rate}\left(\mathcal{C}^{(n)}\right)\xrightarrow{n\to\infty} R$, with $P_b^{(n)}\xrightarrow{n\to \infty} 0$ {(resp. $P_B^{(n)}\xrightarrow{n\to \infty} 0$)}. Hence, any sequence of subcodes of $\left\{\mathcal{C}^{(n)}\right\}_{n\geq 1}$ are also such that their bit-error {and block-error} probabilities go to zero, as the blocklength goes to infinity.

\subsection{Reed-Muller Codes: Definitions and BMS Channel Performance}
\label{sec:rmintro}
We recall the definition of the binary Reed-Muller (RM) family of codes. Codewords of binary RM codes consist of the evaluation vectors of multivariate polynomials over the binary field $\mathbb{F}_2$. Consider the polynomial ring $\mathbb{F}_2[x_1,x_2,\ldots,x_m]$ in $m$ variables. Note that in the specification of a polynomial $f\in \mathbb{F}_2[x_1,x_2,\ldots,x_m]$, only monomials of the form $\prod_{j\in S:S\subseteq [m]} x_j$ need to be considered, since $x^2 = x$ over the field $\mathbb{F}_2$, for an indeterminate $x$. For a polynomial $f\in \mathbb{F}_2[x_1,x_2,\ldots,x_m]$ and a binary vector $\mathbf{z} = (z_1,\ldots,z_m)\in \mathbb{F}_2^m$, we write {$f(\mathbf{z})=f(z_1,\ldots,z_m)$} as the evaluation of $f$ at $\mathbf{z}$. We let the evaluation points be ordered according to the standard lexicographic order on strings in $\mathbb{F}_2^m$, i.e., if $\mathbf{z} = (z_1,\ldots,z_m)$ and $\mathbf{z}^{\prime} = (z_1^{\prime},\ldots,z_m^{\prime})$ are two distinct evaluation points, then, $\mathbf{z}$ occurs before $\mathbf{z}^{\prime}$ in our ordering if and only if for some $i\geq 1$, we have $z_j = z_j^{\prime}$ for all $j<i$, and $z_i < z_i^{\prime}$. Now, let Eval$(f):=\left({f(\mathbf{z})}:\mathbf{z}\in \mathbb{F}_2^m\right)$ be the evaluation vector of $f$, where the coordinates $\mathbf{z}$ are ordered according to the standard lexicographic order. 

\begin{definition}[see \cite{mws}, Chap. 13, or \cite{rm_survey}]
	{For $0\leq r\leq m$}, the $r^{\text{th}}$-order binary Reed-Muller code RM$(m,r)$ is defined as
	\[
	\text{RM}(m,r):=\{\text{Eval}(f): f\in \mathbb{F}_2[x_1,x_2,\ldots,x_m],\ \text{deg}(f)\leq r\},
	\]
	where $\text{deg}(f)$ is the degree of the largest monomial in $f$, and the degree of a monomial $\prod_{j\in S: S\subseteq [m]} x_j$ is simply $|S|$. 
\end{definition}

It is a known fact that the evaluation vectors of all the distinct monomials in the variables $x_1,\ldots, x_m$ are linearly independent over $\mathbb{F}_2$. It then follows that RM$(m,r)$ has dimension $\binom{m}{\le r} := \sum_{i=0}^{r}{m \choose i}$. Furthermore, RM$(m,r)$ has minimum Hamming distance ${d}_{\text{min}}(\text{RM}(m,r))=2^{m-r}$. The weight of a codeword $\mathbf{c} = \text{Eval}(f)$ is the number of $1$s in its evaluation vector, i.e,
\[
\text{wt}\left(\text{Eval}(f)\right):=\left\lvert\{\mathbf{z}\in \mathbb{F}_2^m: f(\mathbf{z})=1\}\right\rvert.
\]

The number of codewords in RM$(m,r)$ of weight $w$, for $w\in[2^{m-r}: 2^m]$, is given by the weight distribution function at $w$:
\[
A_{m,r}(w):=\left\lvert\{\mathbf{c}\in \text{RM}(m,r): \text{wt}\left(\mathbf{c}\right) = w\}\right\rvert.
\]
The subscripts $m$ and $r$ in $A_{m,r}$ will be suppressed when clear from context.

We also set $G_{\text{Lex}}(m,r)$ to be the generator matrix of $\text{RM}(m,r)$ consisting of rows that are the evaluations, in the lexicographic order, of monomials of degree less than or equal to $r$. The columns of $G_{\text{Lex}}(m,r)$ will be indexed by $m$-tuples $\mathbf{b} = (b_1,\ldots,b_m)$ in the lexicographic order.

In this work, we shall use $(d,\infty)$-RLL constrained subcodes of RM codes, over BMS channels. We now recall the main results of Reeves and Pfister in \cite{Reeves} and Abbe and Sandon in \cite{abbesandon}, which provides context to our using RM codes over input-constrained BMS channels.
For a given $R\in (0,1)$, consider any sequence of RM codes $\left\{{\hat{\mathcal{C}}_m} = \text{RM}(m,r_m)\right\}_{m\geq 1}$, under the lexicographic ordering of coordinates, such that rate$\left({\hat{\mathcal{C}}_m}\right) \xrightarrow{m\to \infty}R$. We let $C$ be the capacity of the unconstrained BMS channel under consideration. The following theorems then hold true:
\begin{theorem}[see Theorem 1 of \cite{Reeves} and Theorem 1 of \cite{abbesandon}]
	\label{thm:Reeves}
	Any rate $R\in [0,C)$ is achieved by the sequence of codes $\left\{{\hat{\mathcal{C}}_m}\right\}_{m\geq 1}$, under bit-MAP and {block-MAP} decoding.
\end{theorem}
As an example, {for $R\in (0,1)$,} consider the sequence of RM codes $\left\{{{\mathcal{C}_m}} = \text{RM}(m,r_m)\right\}_{m\geq 1}$ with 
\begin{equation}
	\label{eq:rmval}
	r_m = \max \left\{\left \lfloor \frac{m}{2}+\frac{\sqrt{m}}{2}Q^{-1}(1-R)\right \rfloor,0\right\},
\end{equation}
where $Q(\cdot)$ is the complementary cumulative distribution function (c.c.d.f.) of the standard normal distribution, i.e.,
\[
Q(t) = \frac{1}{\sqrt{2\pi}}\int_{t}^{\infty}e^{-\tau^2/2}d\tau, \ t\in \mathbb{R}.
\]

From Remark 24 in \cite{kud1}, it follows that rate$\left({{\mathcal{C}_m}}\right) \xrightarrow{m\to \infty}R$. Theorem \ref{thm:Reeves} then states that the sequence of codes $\{{{\mathcal{C}_m}}\}_{m\geq 1}$ achieves a rate $R$ over any unconstrained BMS channel, so long as $R<C$.

\begin{remark}
We note that by Theorem \ref{thm:Reeves}, for the unconstrained BEC (resp. unconstrained BSC) with erasure probability $\epsilon \in (0,1)$ (resp. crossover probability $p\in (0,0.5)\cup (0.5,1)$), the sequence of codes $\{\mathcal{C}_m\}_{m\geq 1}$ {with $R = 1-\epsilon-\delta$ (resp. with $R = 1-h_b(p)-\delta)$)} achieves a rate of $1-\epsilon-\delta$ (resp. a rate of $1-h_b(p)-\delta$), for all $\delta>0$ suitably small. 
\end{remark}
The following important property of RM codes, which is sometimes called the Plotkin decomposition (see  \cite[Chap. 13]{mws} and \cite{rm_survey}), will come of use several times in this paper: any Boolean polynomial $f \in \mathbb{F}_2[x_1,\ldots,x_m]$, such that Eval$(f) \in \text{RM}(m,r)$ {(or equivalently, with deg$(f)\leq r$)}, can be expressed as:
\begin{equation}
	\label{eq:rmdecomp}
	f(x_1,\ldots,x_m) = g(x_1,\ldots,x_{m-1}) + x_m\cdot h(x_1,\ldots,x_{m-1}), 
\end{equation}
where $g,h$ are such that Eval$(g)\in \text{RM}(m-1,r)$ and Eval$(h)\in \text{RM}(m-1,r-1)${, and the `$+$' operation is over $\mathbb{F}_2$}. {Note that the polynomials $g$ and $h$ above are uniquely determined, given the polynomial $f$.} Given the sequence of RM codes $\left\{{{\mathcal{C}_m}} = \text{RM}(m,r_m)\right\}_{m\geq 1}$, where $r_m$ is as in \eqref{eq:rmval}, we use the notations ${\mathcal{C}_{m,+}}:=\text{RM}(m-1,r_m)$ and ${\mathcal{C}_{m,-}}:=\text{RM}(m-1,r_m-1)$, with rate$\left({\mathcal{C}_{m,+}}\right):={R_{m,+}}$ and rate$\left({\mathcal{C}_{m,-}}\right):={R_{m,-}}$.

In our $(d,\infty)$-RLL constrained code constructions in this paper, we shall use constrained subcodes of a sequence $\left\{{\hat{\mathcal{C}}_m} = \text{RM}(m,v_m)\right\}_{m\geq 1}$ of rate $R$, and explicitly compute the rates of the coding schemes. From the discussion in Section \ref{sec:lincodesBMS} above, we arrive at the fact that using the bit-MAP or block-MAP decoders of ${\hat{\mathcal{C}}_m}$, the rates of the constrained subcodes of ${\hat{\mathcal{C}}_m}$, computed in this paper, are in fact achievable over $(d,\infty)$-RLL input-constrained BMS channels, so long as $R<C$.

\section{Main Results}
\label{sec:main}
In this section, we briefly state our main theorems, and provide comparisons with the literature. We assume that the BMS channel that we are working with, has an unconstrained capacity of $C\in (0,1)$.
{\subsection{Rates of Subcodes Under the Lexicographic Coordinate Ordering}}
We first fix the coordinate ordering of the RM codes to be the standard lexicographic ordering. Our first approach to designing $(d,\infty)$-RLL constrained codes using RM codes is constructing a sequence of \emph{linear} subcodes of $\{{\mathcal{C}_m} = \text{RM}(m,r_m)\}_{m\geq 1}$, with $r_m$ as in \eqref{eq:rmval}, which respect the $(d,\infty)$-RLL input-constraint, and analyzing the rate of the chosen subcodes. We obtain the following result:

\begin{theorem}
	\label{thm:rm}
	For any $R \in (0,1)$, there exists a sequence of linear $(d,\infty)$-RLL codes $\bigl\{{\mathcal{C}_{m}^{(d)}}\bigr\}_{m\geq 1}$, where ${\mathcal{C}_{m}^{(d)}} \subset {\mathcal{C}_m}$, of rate  $\frac{R}{2^{\left \lceil \log_2(d+1)\right \rceil}}$.
\end{theorem}

The proof of Theorem \ref{thm:rm} (which is Theorem III.2 in \cite{arnk22isit}), which contains an explicit identification of the subcodes $\bigl\{{\mathcal{C}_{m}^{(d)}}\bigr\}_{m\geq 1}$, is provided in Section \ref{sec:rm}. 
From the discussion in Section \ref{sec:rmintro}, we see that by Theorem \ref{thm:rm}, using linear constrained subcodes of RM codes over the $(d,\infty)$-RLL input-constrained BEC, a rate of $\frac{1}{d+1}(1-\epsilon)$ is achievable when $d=2^t-1$, for some $t\in \mathbb{N}$, and a rate of $\frac{1}{2(d+1)}(1-\epsilon)$ is achievable, otherwise. We note, however, that using random coding arguments, or using the techniques in \cite{ZW88} or \cite{arnk20}, a rate of ${\kappa_d}(1-\epsilon)$ is achievable over the  $(d,\infty)$-RLL input-constrained BEC, where ${\kappa_d}$ is the noiseless capacity of the input constraint (for example, ${\kappa_1}\approx 0.6942$ and ${\kappa_2}\approx 0.5515$). 
For the $(d,\infty)$-RLL input-constrained BSC, similarly, a rate of $\frac{1}{d+1}(1-h_b(p))$ is achievable when $d=2^t-1$, for some $t\in \mathbb{N}$, and a rate of $\frac{1}{2(d+1)}(1-h_b(p))$ is achievable, otherwise. Such a result is in the spirit of, but is weaker than, the conjecture by Wolf \cite{Wolf} that a rate of ${\kappa_d}(1-h_b(p))$ is achievable over the $(d,\infty)$-RLL input-constrained BSC.

For the specific case when $d=1$, we now state an existence result that provides another lower bound on rates of (potentially non-linear) $(1,\infty)$-RLL constrained subcodes of RM codes of rate $R$.

\begin{theorem}
	\label{thm:nonlin}
	For any $R\in (0,1)$ and for any sequence of codes $\left\{{\hat{\mathcal{C}}_m} = \text{RM}(m,r_m)\right\}_{m\geq 1}$ of rate $R$, following the lexicographic coordinate ordering, there exists a sequence of {$(1,\infty)$-RLL subcodes} of rate at least $\max\left(0,R-\frac38\right)$.
\end{theorem}

The proof of Theorem \ref{thm:nonlin} is provided in Section \ref{sec:rm}. Again, we note from the discussion in Section \ref{sec:rmintro}, that Theorem \ref{thm:nonlin} implies that rates of at least $\max\left(0,C-\frac38\right)$ are achievable over the $(d,\infty)$-RLL input-constrained BMS channel. Further, we observe that the lower bound in the theorem above beats the achievable rate of $\frac{R}{2}$ in Theorem \ref{thm:rm}, when $R>0.75$. 

Next, we state a theorem that provides upper bounds on the largest rate of \emph{linear} $(d,\infty)$-RLL subcodes of RM codes, where the coordinates are ordered according to the lexicographic ordering. Fix any sequence of codes $\left\{{\hat{\mathcal{C}}_m} = \text{RM}(m, r_m)\right\}_{m\geq 1}$ of rate $R\in (0,1)$ and let ${\overline{\mathsf{R}}^{(d)}(\hat{\mathcal{C}})}$ be the largest rate of linear $(d,\infty)$-RLL subcodes of $\left\{{\hat{\mathcal{C}}_m}\right\}_{m\geq 1}$. Formally,
\begin{equation}
	\label{eq:Rub}
	{\overline{\mathsf{R}}^{(d)}(\hat{\mathcal{C}})}:=\limsup_{m\to \infty} \max_{{\overline{\mathcal{H}}^{(d)}}\subseteq {\hat{\mathcal{C}}_m}} \frac{\log_2\left \lvert {\overline{\mathcal{H}}^{(d)}}\right\rvert}{2^m},
\end{equation}
where the maximization is over linear $(d,\infty)$-RLL subcodes ${\overline{\mathcal{H}}^{(d)}}$ of ${\hat{\mathcal{C}}_m}$. Then,

\begin{theorem}
	\label{thm:rmlinub}
	For any $R\in (0,1)$ and for any sequence of codes $\left\{{\hat{\mathcal{C}}_m} = \text{RM}(m,r_m)\right\}_{m\geq 1}$ {of rate $R$}, following the lexicographic coordinate ordering,
	\[
	{\overline{\mathsf{R}}^{(d)}(\hat{\mathcal{C}})} \leq \frac{R}{d+1}.
	\]
\end{theorem}
Thus, Theorem \ref{thm:rmlinub} shows that the sequence of simple linear subcodes $\left\{{\mathcal{C}_{m}^{(d)}}\right\}_{m\geq 1}$, identified in Theorem \ref{thm:rm}, is rate-optimal whenever $d+1$ is a power of $2$, in that it achieves the rate upper bound of $\frac{R}{d+1}$. Moreover, it follows from Theorem \ref{thm:rmlinub}, that the subcodes in Theorem \ref{thm:nonlin} must be \emph{non-linear} when $R>0.75$. Theorem \ref{thm:rmlinub} is proved in Section \ref{sec:rmlinub}. Also, from the discussion in Section \ref{sec:notation} and from {Theorem} \ref{thm:rmlinub}, we see that the largest rate achievable over a $(d,\infty)$-RLL input-constrained BMS channel, using linear $(d,\infty)$-RLL subcodes of RM codes, when the sub-optimal bit-MAP {or block-MAP} decoders of the larger RM codes are used, is bounded above by $\frac{C}{d+1}$.

We remark here that the problem of identifying linear codes that are subsets of the set of $(d,\infty)$-RLL sequences of a fixed length, has been studied in \cite{lechner}. The results therein show that the largest linear code within $S_{(d,\infty)}^{(m)}$ has rate no larger than $\frac{1}{d+1}$, as $m\to \infty$. However, such a result offers no insight into rates achievable over BMS channels.

Concluding the discussion on rates achievable using {subcodes} of RM codes, following the lexicographic coordinate ordering, we state a theorem that provides an upper bound on the largest rate of $(1,\infty)$-RLL subcodes of the sequence $\{{\mathcal{C}_m} = \text{RM}(m,r_m)\}_{m\geq 1}$, where $r_m$ is as in \eqref{eq:rmval}.
We formally define this largest rate to be 
\[
{\mathsf{R}^{(1)}(\mathcal{C})}:=\limsup_{m\to \infty} \max_{{\mathcal{H}^{(1)}}\subseteq {\mathcal{C}_m}}\frac{\log_2\left\lvert{\mathcal{H}^{(1)}}\right\rvert}{2^m},
\]
where the maximization is over $(1,\infty)$-RLL subcodes ${\mathcal{H}^{(1)}}$ of ${\mathcal{C}_m}$. Then,
\begin{theorem}
	\label{thm:rmub}
	For the sequence of codes $\{{{\mathcal{C}_m}} = \text{RM}(m,r_m)\}_{m\geq 1}$, with $r_m$ as in \eqref{eq:rmval},
	\[
	{\mathsf{R}^{(1)}(\mathcal{C})}\leq \min\left(\frac{7R}{8},{\kappa_1}\right),
	\]
	where ${\kappa_1} = \log_2\left(\frac{1+\sqrt{5}}{2}\right)\approx 0.6942$ is the noiseless capacity of the $(1,\infty)$-RLL constraint.
\end{theorem}
The proof of the theorem is taken up in Section \ref{sec:rmub}. We note that the upper bound in the theorem above is an improvement over the upper bound in \cite{arnk22isit}, due to the use of better upper bounds on the weight distribution, from \cite{anuprao}, as compared to those in \cite{sam}. Figure \ref{fig:ubplot} shows a comparison between the upper bound in Theorem \ref{thm:rmub}, and the lower bounds of Theorems \ref{thm:rm} and \ref{thm:nonlin}, for the case when $d=1$. From the discussion in Section \ref{sec:rmintro}, we see that Theorem \ref{thm:rmub} shows that the rate achievable over $(1,\infty)$-RLL input-constrained BMS channels, using constrained subcodes of $\{{\mathcal{C}_m}\}_{m\geq 1}$, when the sub-optimal bit-MAP or {block-MAP} decoders of the larger RM codes are used, is bounded above by $\min\left(\frac{7C}{8},{\kappa_1}\right)$, where $C$ is the capacity of the unconstrained BMS channel. Figure \ref{fig:ubcompare} shows comparisons, for the specific case of the $(1,\infty)$-RLL input-constrained BEC, of the upper bound of $\min\left(\frac{7}{8}\cdot (1-\epsilon),{\kappa_1}\right)$, obtained by sub-optimal decoding, in Theorem \ref{thm:rmub}, with the achievable rate of ${\kappa_1}\cdot(1-\epsilon)$ (from \cite{Hanerasure} and \cite{arnk20}), and the numerically computed achievable rates using the Monte-Carlo method in \cite{arnold} (or the stochastic approximation scheme in \cite{arnk21ncc}). For large values of the erasure probability $\epsilon$, we observe that the upper bound of $\min\left(\frac{7}{8}\cdot (1-\epsilon),{\kappa_1}\right)$ lies below the achievable rates of \cite{arnold}, thereby indicating that it is not possible to achieve the capacity of the $(1,\infty)$-RLL input-constrained BEC, using $(1,\infty)$-RLL subcodes of $\{{\mathcal{C}_m}\}_{m\geq 1}$, when the bit-MAP or block-MAP decoders of $\{{\mathcal{C}_m}\}_{m\geq 1}$ are used for decoding. We conjecture that this numerically verified fact is indeed true.

{\subsection{Rates of Subcodes Under Alternative Coordinate Orderings}}
Next, we consider situations where the coordinates of the RM codes follow orderings different from the standard lexicographic ordering. First, we study upper bounds on the rates of \emph{linear} $(d,\infty)$-RLL subcodes of RM codes, whose coordinates are ordered according to a Gray ordering (see Section \ref{sec:perm} for a description of a Gray ordering). For a fixed $R\in (0,1)$, let $\left\{{\mathcal{C}_m^\text{G}}\right\}_{m\geq 1}$ be any sequence of RM codes following a Gray coordinate ordering, such that rate$\left({\mathcal{C}_m^\text{G}}\right)\xrightarrow{m\to \infty} R$. Let ${\overline{\mathsf{R}}^{(d)}(\mathcal{C}^\text{G})}$ denote the largest rate of \emph{linear} $(d,\infty)$-RLL subcodes of $\left\{{\mathcal{C}_m^\text{G}}\right\}_{m\geq 1}$. Formally,
\begin{equation}
	\label{eq:Rubgray}
	{\overline{\mathsf{R}}^{(d)}(\mathcal{C}^\text{G})}:=\limsup_{m\to \infty} \max_{{\overline{\mathcal{H}}_{\text{G}}^{(d)}}\subseteq {\mathcal{C}_m^{\text{G}}}}\frac{\log_2\left \lvert {\overline{\mathcal{H}}_{\text{G}}^{(d)}}\right\rvert}{2^m},
\end{equation}
where the maximization is over linear $(d,\infty)$-RLL subcodes ${\overline{\mathcal{H}}_{\text{G}}^{(d)}}$ of $\mathcal{C}_m^{\text{G}}(R)$.
We obtain the following result:

\begin{theorem}
	\label{thm:grayinf}
	For any $R\in (0,1)$ and for any sequence of RM codes {$\left\{{\mathcal{C}_m^\text{G}}\right\}_{m\geq 1}$ of rate $R$, following a Gray coordinate ordering,}
	\[
	{\overline{\mathsf{R}}^{(d)}(\mathcal{C}^\text{G})}\leq \frac{R}{d+1}.
	\]
\end{theorem}
The proof of Theorem \ref{thm:grayinf} is provided in Section \ref{sec:perm}. Again, Theorem \ref{thm:grayinf} provides an upper bound on the rates achieved over a $(d,\infty)$-RLL input-constrained BMS channel, using {linear} subcodes of Gray-ordered RM codes of rate $R$, when the sub-optimal bit-MAP decoders of $\left\{{\mathcal{C}_m^\text{G}}\right\}_{m\geq 1}$ are used, for decoding.

Now, we consider arbitrary orderings of coordinates, defined by the sequence of permutations $\left(\pi_m\right)_{m\geq 1}$, with $\pi_m: [0:2^m-1]\to [0:2^m-1]$. As with the Gray coordinate ordering, we define the sequence of $\pi$-ordered RM codes $\left\{{\mathcal{C}_m^\pi}\right\}_{m\geq 1}$, with
\begin{align*}{\mathcal{C}_{m}^\pi}:= \big\{(c_{\pi_m(0)},c_{\pi_m(1)},\ldots,c_{\pi_m({2^m}-1)}): (c_0,c_1,\ldots,c_{{2^m}-1})\in {\hat{\mathcal{C}}_m}\big\},\end{align*}
where $\left\{{\hat{\mathcal{C}}_m} = \text{RM}(m,r_m)\right\}_{m\geq 1}$ is any sequence of RM codes of rate $R$. We also define ${\overline{\mathcal{H}}_{\pi}^{(d)}}$ be the largest \emph{linear} $(d,\infty)$-RLL subcode of ${\mathcal{C}_m^\pi}$. The theorem below is then shown to hold:
\begin{theorem}
	\label{thm:genperminf}
	For any $R\in (0,1)$, for large $m$ and for all but a vanishing fraction of coordinate permutations, $\pi_m: [0:2^m-1]\to [0:2^m-1]$, the following rate upper bound holds:
	\[
	\frac{\log_2\left \lvert {\overline{\mathcal{H}}_{\pi}^{(d)}}\right\rvert}{2^m}\leq \frac{R}{d+1}+\epsilon_m,
	\]
	where $\epsilon_m\xrightarrow{m\to \infty} 0$.
\end{theorem}
Section \ref{sec:perm} contains the proof of Theorem \ref{thm:genperminf}.
{\subsection{Rates Using Other Coding Strategies}}
\begin{figure*}[!t]
	\centering
	\includegraphics[width = 0.8\textwidth]{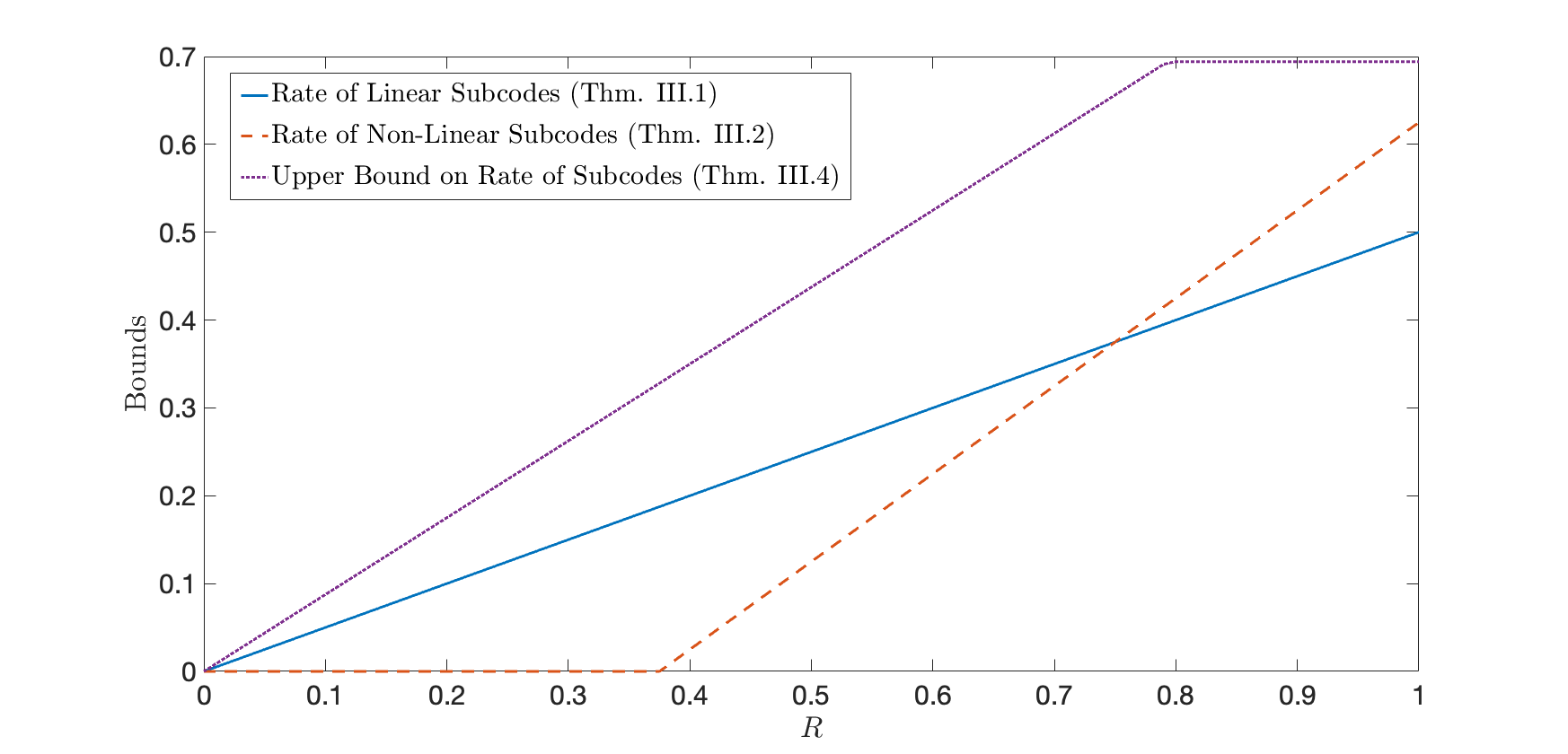}
	\caption{A comparison between the upper bound of Theorem \ref{thm:rmub} and achievable rates of $R/2$ and $\max\left(0,R-\frac38\right)$, from Theorems III.1 and III.2, respectively,  when $d=1$.}
	\label{fig:ubplot}
\end{figure*}

\begin{figure*}[!t]
	\centering
	\includegraphics[width = 0.8\textwidth]{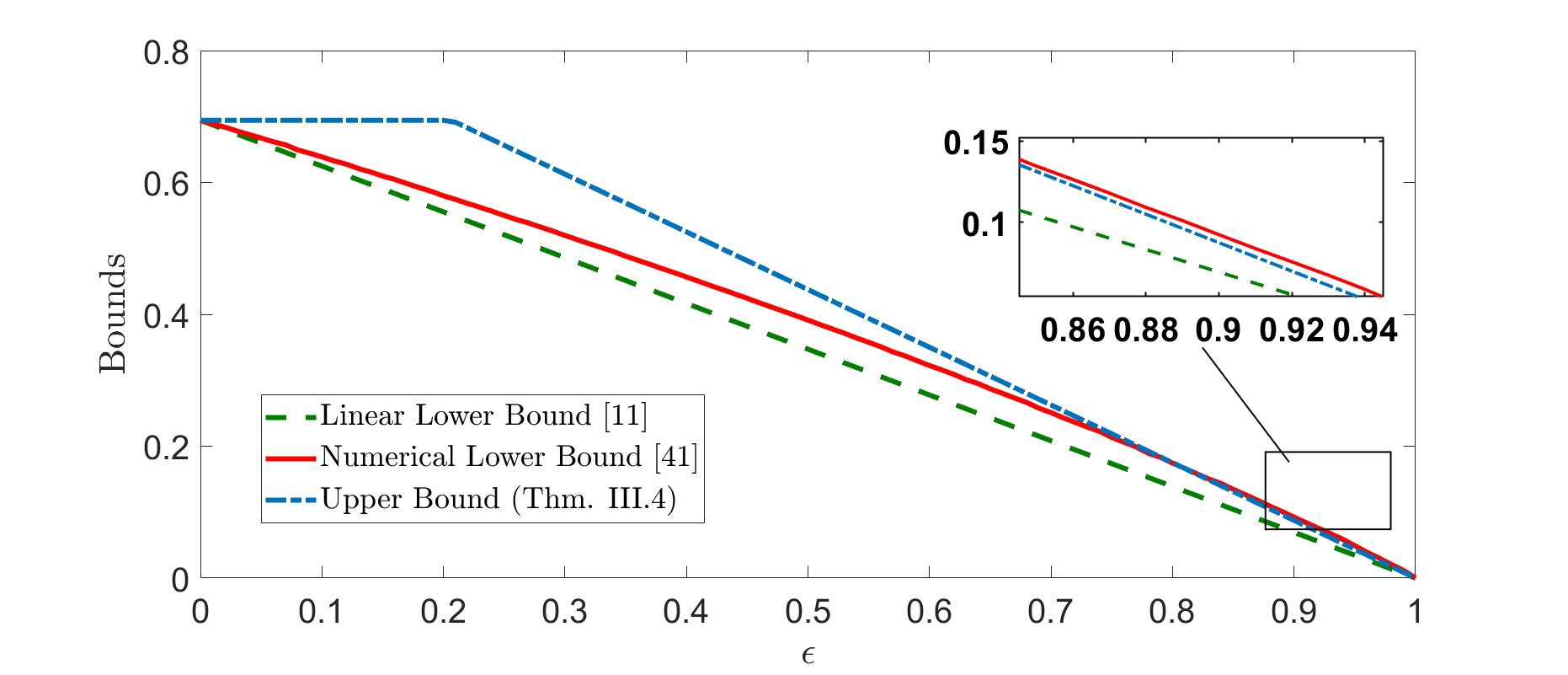}
	\caption{A comparison between the upper bound of Theorem \ref{thm:rmub} and the achievable rates of \cite{Hanerasure} and \cite{arnold}. For large $\epsilon$, the upper bound of Theorem \ref{thm:rmub}, under sub-optimal decoding, lies below the numerically-computed achievable rates in \cite{arnold}.}
	\label{fig:ubcompare}
\end{figure*}

Finally, we design $(d,\infty)$-RLL constrained codes, whose rates improve on those in Theorems \ref{thm:rm} and \ref{thm:nonlin}, by using a two-stage (or concatenated) encoding procedure that employs systematic RM codes of rate $R$. 

\begin{theorem}
	\label{thm:rmcosets}
	For any $R\in (0,1)$, there exists a sequence of $(d,\infty)$-RLL constrained concatenated codes $\left\{{\mathcal{C}}_m^{\text{conc}}\right\}_{m\geq 1}$, constructed using systematic RM codes of rate $R$, such that
	\[
	\liminf_{m\to \infty} \text{rate}(\mathcal{C}_m^{\text{conc}})\geq \frac{{\kappa_d}\cdot R^2\cdot 2^{-\left \lceil \log_2(d+1)\right \rceil}}{R^2\cdot 2^{-\left \lceil \log_2(d+1)\right \rceil} + 1-R+2^{-\tau}},
	\]
	where $\tau$ is an arbitrarily large, but fixed, positive integer. Further, the rate lower bound above is achievable, {under block-MAP decoding}, over a $(d,\infty)$-RLL input-constrained BMS channel, when $R<C$, where $C$ is the capacity of the unconstrained BMS channel.
\end{theorem}
It can be checked that the rates achieved using Theorem \ref{thm:rmcosets} are better than those achieved using Theorem \ref{thm:rm} (and in fact, better than those achieved using any sequence of linear $(d,\infty)$-RLL subcodes of RM codes) and Theorem \ref{thm:nonlin}, for low noise regimes of the BMS channel.  For example, when $d=1$, the rates achieved using the codes in Theorem \ref{thm:rmcosets} are better than those achieved using linear subcodes, for $R\gtrapprox 0.7613$, and are better than those achieved using the subcodes of Theorem \ref{thm:nonlin}, for $R\lessapprox 0.55$ and $R\gtrapprox 0.79$. Figures \ref{fig:first} and \ref{fig:second} show comparisons between the lower bounds (achievable rates), {under block-MAP decoding,} in Theorems \ref{thm:rm}, \ref{thm:nonlin}, and \ref{thm:rmcosets}, with the coset-averaging bound of \cite{pvk}, for $d=1$ and $d=2$, respectively. While \cite{pvk} provides existence results on rates achieved using cosets of RM codes, with the rates calculated therein being better than those in Theorem \ref{thm:rmcosets} in the low noise regimes of the BMS channel, our construction is more explicit. The code construction leading to Theorem \ref{thm:rmcosets}, and the proof of achievability of the rate lower bound, is taken up in Section \ref{sec:cosets}. We mention here that Theorem \ref{thm:rmlinub} and Theorems \ref{thm:grayinf}--\ref{thm:rmcosets} first appeared in \cite{arnk22itw}.

\begin{figure*}%
	\centering	
	\includegraphics[width=0.8\textwidth]{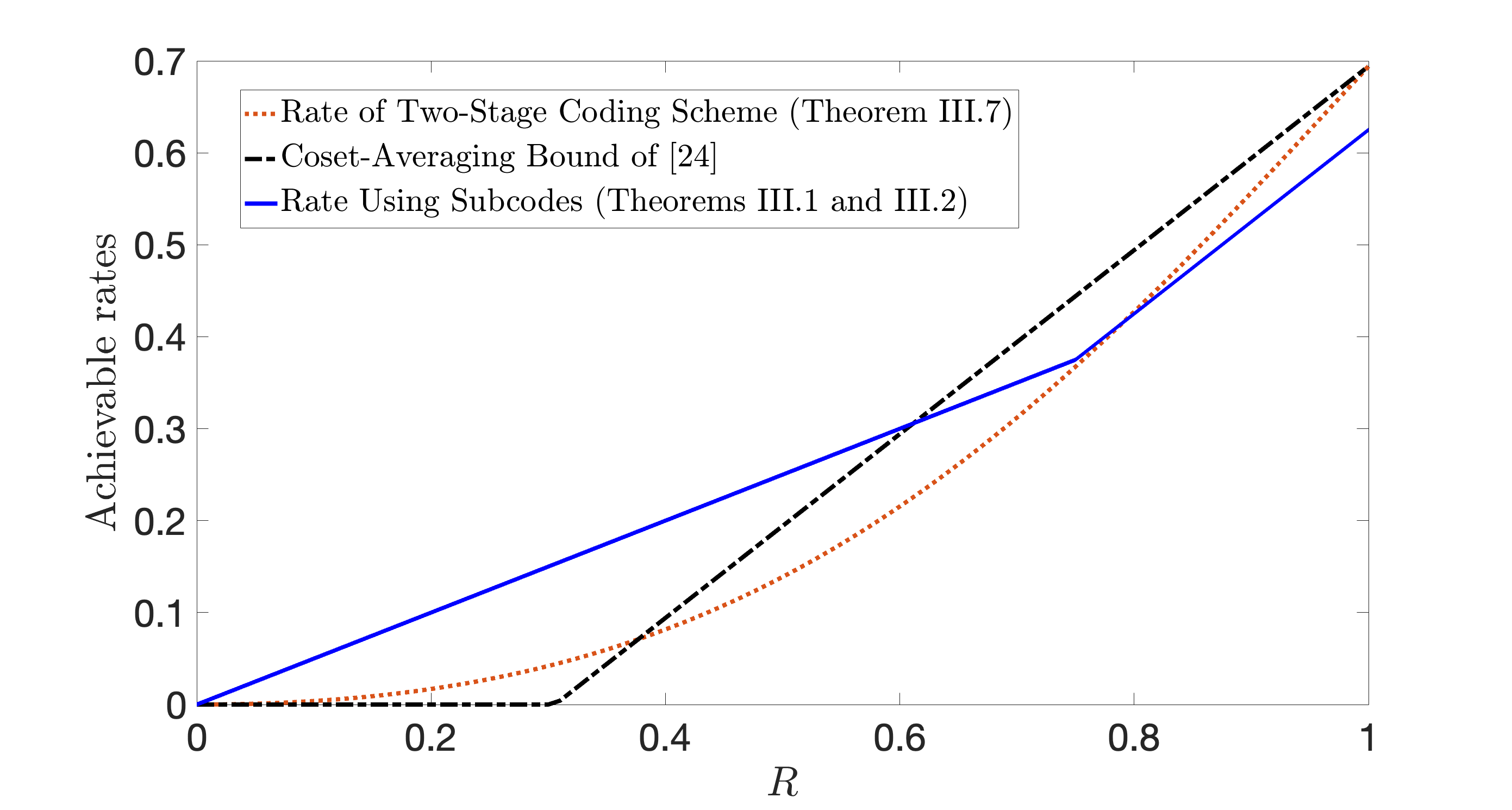}%
	\caption{Plot comparing, for $d=1$, the rate lower bounds of $\max\left(\frac{R}{2},R-\frac38\right)$ achieved using subcodes, from Theorems \ref{thm:rm} and \ref{thm:nonlin}, with the rate lower bound achieved using Theorem \ref{thm:rmcosets}, with $\tau = 50$, and the coset-averaging lower bound of $\max(0,{\kappa_1}+R-1)$, of \cite{pvk}. Here, the noiseless capacity, ${\kappa_1} \approx 0.6942$.}
	\label{fig:first}%
\end{figure*}

We end this section with a remark. Note that the all-ones codeword $\mathbf{1}$ belongs to the RM code. Since any codeword $\mathbf{c}$ that respects the $(0,1)$-RLL constraint can be written as $\mathbf{c} = \mathbf{1}+\mathbf{\hat{c}}$, where $\mathbf{\hat{c}}$ respects the $(1,\infty)$-RLL constraint, the lower and upper bounds of the theorems above hold for the rate of $(0,1)$-RLL subcodes as well. Moreover, since for any $k>1$, a $(0,1)$-RLL subcode of an RM code is a subset of a $(0,k)$-RLL subcode, the lower bounds of Theorems \ref{thm:rm}, \ref{thm:nonlin}, and \ref{thm:rmcosets}, hold over $(0,k)$-RLL input-constrained BMS channels as well.
\section{Achievable Rates Using Subcodes}
\label{sec:rm}

As mentioned in Section \ref{sec:main}, we work with the Reed-Muller (RM) family of codes, $\{{\mathcal{C}_m} = \text{RM}(m,r_m)\}_{m\geq 1}$, under the lexicographic coordinate ordering, with
\begin{equation*}
	r_m = \max \left\{\left \lfloor \frac{m}{2}+\frac{\sqrt{m}}{2}Q^{-1}(1-R)\right \rfloor,0\right\},
\end{equation*}
and rate $R\in (0,1)$. We then select subcodes of these codes that respect the $(d,\infty)$-RLL constraint, and compute their rate. 
We first consider the case when $d=1$. For this situation, we provide a complete characterization of $(1,\infty)$-RLL constrained subcodes of RM codes, which will help us identify $(1,\infty)$-RLL constrained subcodes of $\{{\mathcal{C}_m}\}_{m\geq 1}$. This characterization will also come of use in our derivation of upper bounds on the rates of general (potentially non-linear) $(1,\infty)$-RLL subcodes of RM codes, in Section \ref{sec:rmub}. 

\begin{figure*}%
	\centering

	\includegraphics[width=0.8\textwidth]{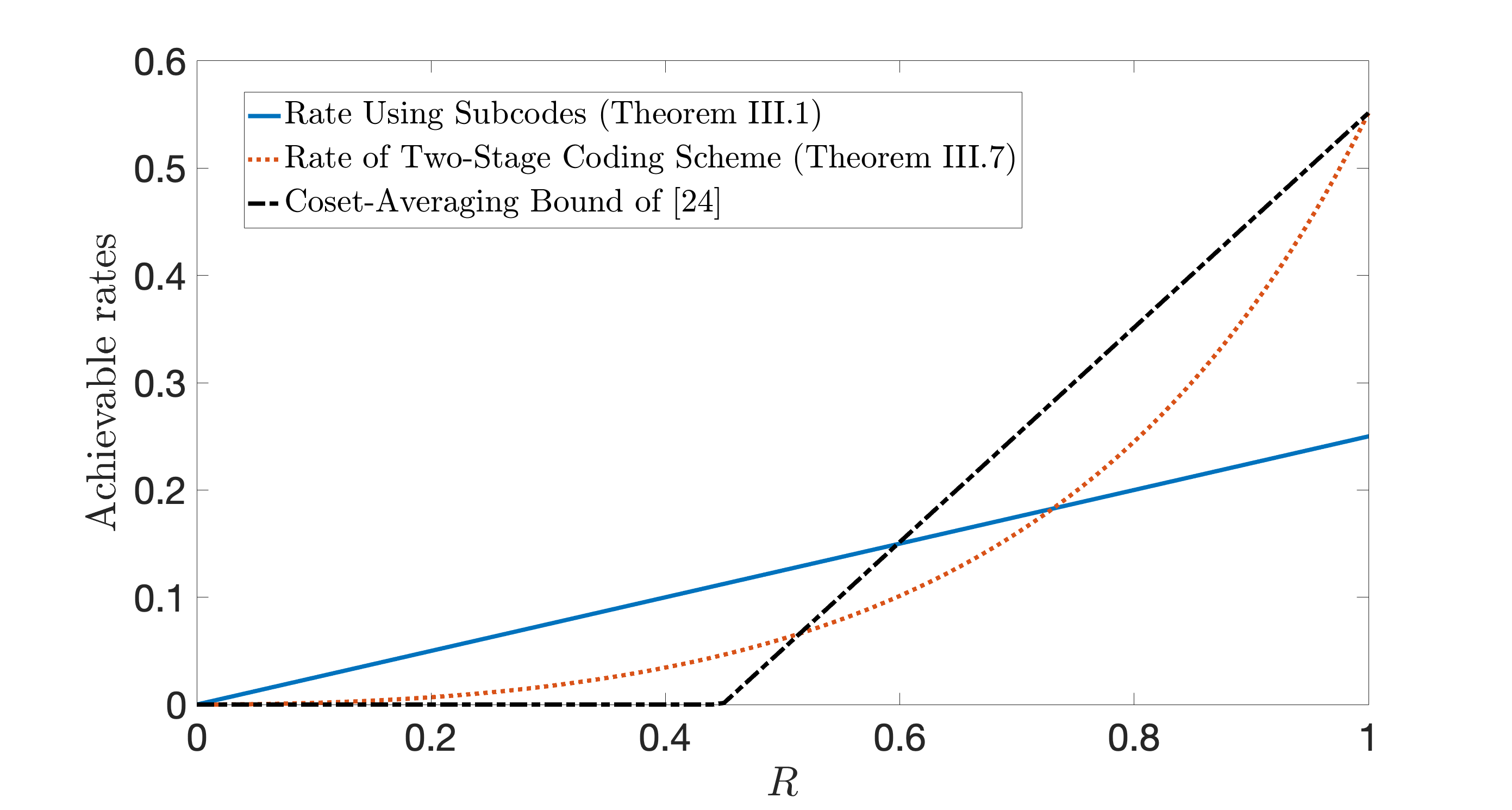}%
	\caption{Plot comparing, for $d=2$, the rate lower bound of $R/4$ achieved using subcodes, from Theorem \ref{thm:rm}, the rate lower bound achieved using Theorem \ref{thm:rmcosets}, with $\tau = 50$, and the coset-averaging lower bound of $\max(0,{\kappa_2}+R-1)$, of \cite{pvk}. Here, the noiseless capacity, ${\kappa_2} \approx 0.5515$.}
	\label{fig:second}%
\end{figure*}

We now set up some definitions and notation for our characterization:  we define a ``run'' of coordinates belonging to a set $\mathcal{A}\in \{0,1\}^m$, to be a contiguous collection of coordinates, $\left(i,i+1,\ldots,i+\ell-1\right)$, {for some integers $\ell \ge 1$ and $i \in [0:2^m-\ell]$}, such that $\mathbf{B}(j)\in \mathcal{A}$, for all $i\leq j\leq i+\ell-1$, and $\mathbf{B}(i-1),\mathbf{B}(i+\ell)\notin \mathcal{A}$. {To cover the corner cases of $i=0$ and $i=2^m-\ell$,} we set $\mathbf{B}(i-1)$ and $\mathbf{B}(i+\ell)$, respectively, to be the dummy symbol $\times$, which does not belong to $\mathcal{A}$. The length of such a run is $\ell$; {note that we allow $\ell$ to be $1$}. For example, a run of $1$s in a vector $\mathbf{v}\in \{0,1\}^m$ is a collection of contiguous coordinates, $(i,\ldots,i+\ell-1)$, such that $v(i+k) = 1$, for $0\leq k\leq \ell-1$. Finally, given the code RM$(m,r)$, for $g$ as in equation \eqref{eq:rmdecomp}, we let {$\Gamma(g)$ denote the set of all coordinates $\mathbf{b}$, excluding the coordinate $(1,1,\ldots,1)$, such that $\mathbf{b}$ is the last coordinate in a run of $0$s in Eval$(g)$, i.e.,}
\[
\Gamma(g):= \{\boldsymbol{b} = (b_1,\ldots,b_{m-1}):\ \boldsymbol{b} \text{ is the end of a run of $0$s in Eval$(g)$,} \text{ $\boldsymbol{b}\neq (1,1,\ldots,1)$}\}.
\]
We now present our characterization of $(1,\infty)$-RLL subcodes of the code RM$(m,r)$:
\begin{proposition}
	\label{lem:necesscond}
	For any Eval$(f)\in \text{RM}(m,r)$, we have that Eval$(f)\in S_{(1,\infty)}$ if and only if the following two conditions (C1) and (C2) are simultaneously satisfied:
	\begin{align*}
		&\text{(C1):  supp$(\text{Eval}(g)) \subseteq \text{supp}(\text{Eval}(h))$}\\
		&\text{(C2): $h(\boldsymbol{b}) = 0$, if $\boldsymbol{b}\in \Gamma(g)$},
	\end{align*}
	where $g,h$ are as in \eqref{eq:rmdecomp}.
\end{proposition}
\begin{proof}
	First, we shall prove that if Eval$(f)\in S_{(1,\infty)}$, then (C1) and (C2) hold.
	
	To show that (C1) must hold, let us assume the contrary. Suppose that there exists some evaluation point $\mathbf{b} = (b_1,\ldots,b_{m-1}) \in \mathbb{F}_2^{m-1}$ such that $g(\mathbf{b})=1$ and $h(\mathbf{b})=0$. Then it follows that at evaluation points $\mathbf{b}_1,\mathbf{b}_2 \in \mathbb{F}_2^m$ such that $\mathbf{b}_1 = \mathbf{b}0$ and $\mathbf{b}_2 = \mathbf{b}1$, we have $f(\mathbf{b}_1) = f(\mathbf{b}_2) = 1$. Since by construction the points $\mathbf{b}_1$ and $\mathbf{b}_2$ are consecutive in the lexicographic ordering, the codeword Eval$(f)\notin S_{(1,\infty)}$. 
	
	Similarly, to show that (C2) must hold, we assume the contrary. Suppose that there exists some evaluation point $\mathbf{b} = {\mathbf{B}_{m-1}}(i) \in \mathbb{F}_2^{m-1}$, for some $i \in [0:2^{m-1}-2]$, such that $\mathbf{b} \in \Gamma(g)$ and $h(\mathbf{b})=1$. We let $\mathbf{b}^\prime = (b_1^\prime,\ldots,b_{m-1}^\prime)$ be such that $\mathbf{b}^\prime = {\mathbf{B}_{m-1}}(i+1)$. We note that in the code RM$(m,r)$, the coordinates $\mathbf{b}_1= \mathbf{b}1$ and $\mathbf{b}_2=\mathbf{b}^\prime0$ occur successively. Again, by simply evaluating \eqref{eq:rmdecomp} at the coordinates $\mathbf{b}_1$ and $\mathbf{b}_2$, we obtain that $f(\mathbf{b}_1) = f(\mathbf{b}_2) = 1$, implying that the codeword Eval$(f)\notin S_{(1,\infty)}$. 
	
	Next, we shall prove the converse, i.e., that if (C1) and (C2) hold, then Eval$(f) \in S_{(1,\infty)}$. Pick any pair of consecutive coordinates $\mathbf{z}_1 = {\mathbf{B}_m}(i)$ and $\mathbf{z}_2 = {\mathbf{B}_m}(i+1)$, in the lexicographic ordering, for some $i\in [0:2^m-2]$. Note that it suffices to prove that if (C1) and (C2) hold, then it cannot be that $f(\mathbf{z}_1) = f(\mathbf{z}_2) = 1$. Now, consider the following two cases, for $\mathbf{b}, \mathbf{b}^\prime \in \{0,1\}^{m-1}$:
	\begin{enumerate}
		\item $\mathbf{z}_1 = \mathbf{b}0$ \textbf{and} $\mathbf{z}_2 = \mathbf{b}1$: In this case, if $f(\mathbf{z}_1) = f(\mathbf{b}0) = 1$, then, from the Plotkin decomposition in \eqref{eq:rmdecomp}, we see that $g(\mathbf{b}) = 1$. From (C1), it hence holds that $h(\mathbf{b}) = 1$. Thus, $f(\mathbf{z}_2) = f(\mathbf{b}1) = g(\mathbf{b}) + h(\mathbf{b}) = 0$.
		\item $\mathbf{z}_1 = \mathbf{b}1$ \textbf{and} $\mathbf{z}_2 = \mathbf{b}^\prime0$: Suppose that $f(\mathbf{z}_1) = f(\mathbf{b}1) = 1$. This then implies that $\mathbf{b}\notin \Gamma(g)$, since otherwise, $f(\mathbf{z}_1) = g(\mathbf{b})+h(\mathbf{b}) = 0$, by (C2). Further, it cannot be that $g(\mathbf{b}) = 1$, as then, using (C1), we see that $f(\mathbf{z}_1) = g(\mathbf{b})+h(\mathbf{b}) = 0$. Hence, it must be that $g(\mathbf{b}) = 0$ and $g(\mathbf{b}^\prime) = 0$, in which case, it immediately follows that $f(\mathbf{z}_2) = g(\mathbf{b}^\prime) = 0$.
	\end{enumerate}%
\end{proof}

{\subsection{Construction of Linear $(d,\infty)$-RLL Constrained Subcodes}}
Given the characterization in Proposition \ref{lem:necesscond}, our first construction of a \emph{linear} $(1,\infty)$-RLL subcode of ${\mathcal{C}_m}$ is simply to pick those codewords Eval$(f) \in {\mathcal{C}_m}$ such that $g\equiv 0$, where $g$ is as in equation \eqref{eq:rmdecomp}. It is straightforward to verify that both (C1) and (C2) in Proposition \ref{lem:necesscond} are trivially satisfied.

In other words, we define the $(1,\infty)$-RLL constrained subcode ${\mathcal{C}_m^{(1)}}$ of ${\mathcal{C}_m}$ to be
\begin{equation}
	\label{eq:C_1inf}
	{\mathcal{C}_m^{(1)}}:=\Big\{\text{Eval}(f): f = x_m\cdot h(x_{1},\ldots, x_{m-1}), \text{ where } \text{deg}(h)\leq r_m-1\Big\}.
\end{equation}

Towards computing the rates of subcodes we work with in this paper, we state and prove the following lemma:

\begin{lemma}
	\label{lem:rate}
	For $r_m$ as defined in \eqref{eq:rmval} and any sequence of positive integers $(t_m)_{m\geq 1}$ such that $t_m = o(\!\!\sqrt{m})$, we have
	$$
	\lim_{m \to \infty} \frac{1}{2^{m-t_m}} \binom{m-t_m}{\le r_m} \ = \ R.
	$$
	In particular, for any fixed integer $t > 0$, $\lim\limits_{m \to \infty} \frac{1}{2^m} \binom{m-t}{\le r_m} = 2^{-t}R$.
\end{lemma}

\begin{proof}
	Let $S_m$ denote a Bin$(m,\frac12)$ random variable, and note that $\frac{1}{2^{m-t_m}} \binom{m-t_m}{\le r_m}$ equals $\Pr[S_{m-t_m} \le r_m]$. Further, note that by our choice of $r_m$, for any integer $t > 0$, we have for all $m$ large enough,
	\begin{align}
		|r_m - r_{m-t}| &\leq \biggl| \frac{m}{2} + \frac{\sqrt{m}}{2}Q^{-1}(1-R) \, - \  \left(\frac{m-t}{2} + \frac{\sqrt{m-t}}{2}Q^{-1}(1-R)\right) + 1 \biggr| \notag\\
		&\leq \frac{t}{2}+\frac{\sqrt{t}}{2}\lvert Q^{-1}(1-R)\rvert + 1. \label{eq:rm_diff}
	\end{align}
	Hence, we have $r_{m-t_m} - \nu_m \le r_m \le r_{m-t_m} + \nu_m$, with $\nu_m := \frac{t_m}{2}+\frac{\sqrt{t_m}}{2}\lvert Q^{-1}(1-R)\rvert + 1$.
	Consequently, $\Pr[S_{m-t_m} \le r_{m-t_m} - \nu_m] \le \Pr[S_{m-t_m} \le r_m] \le \Pr[S_{m-t_m} \le r_{m-t_m} + \nu_m]$. Setting $\overline{S}_{m-t_m} := \frac{S_{m-t_m} - \frac12(m-t_m)}{\frac12\sqrt{m-t_m}}$, we have 
	\begin{align}
		& \Pr[\overline{S}_{m-t_m} \le Q^{-1}(1-R) - \frac{\nu_m}{\frac12\sqrt{m-t_m}}] 
 \ \le \ \Pr[S_{m-t_m} \le r_m] 
\ \le \ \Pr[\overline{S}_{m-t_m} \le Q^{-1}(1-R) + \frac{\nu_m}{\frac12\sqrt{m-t_m}}].
		\label{sandwich}
	\end{align}
	
	Now, by the central limit theorem (or, in this special case, by the de Moivre-Laplace theorem), $\overline{S}_{m-t_m}$ converges in distribution to a standard normal random variable, $Z$. Therefore, via \eqref{sandwich} and the fact that $t_m$ and $\nu_m$ are both $o(\!\!\sqrt{m})$, we obtain that 
	$$
	\lim_{m \to \infty} \Pr[S_{m-t_m} \le r_m]  \ = \ \Pr[Z \le Q^{-1}(1-R)] \ = \ R,
	$$
	which proves the lemma.
\end{proof}

From Lemma \ref{lem:rate}, we calculate the rate of the $(1,\infty)$-RLL constrained subcode ${\mathcal{C}_m^{(1)}}$ in \eqref{eq:C_1inf} as:
\begin{align}
\text{rate}\left({\mathcal{C}_m^{(1)}}\right) \ & =\ \frac{\log_2\left(\left\lvert{\mathcal{C}_m^{(1)}}\right \rvert\right)}{2^m} \notag\\ 
	&= \ \frac{{m-1 \choose \leq r_m-1}}{2^m} 
	\ = \ \frac{\binom{m-1}{\le r_m-1}}{\binom{m-1}{\le r_m}} \, \frac{{m-1 \choose \leq r_m}}{2^m}
	\ \xrightarrow{m \to \infty} \ \frac{R}{2}. \label{eq:rate_1inf}
\end{align}

We now extend our simple construction of linear $(1,\infty)$-RLL subcodes in \eqref{eq:C_1inf} and the rate computation in \eqref{eq:rate_1inf} to general $d$, thereby proving Theorem \ref{thm:rm}.
But before we do so, we state a simple observation, presented below as a lemma. We recall the definition of the support of a vector $\mathbf{c}\in \mathbb{F}_2^{n}$:
$
\text{supp}(\mathbf{c}) = \{i: c_i = 1\}.
$

\begin{lemma}
	\label{lem:simplerm}
	Given $d\geq 1$, if $\hat{\mathbf{c}}$ is such that $\hat{\mathbf{c}}\in S_{(d,\infty)}$, and supp$(\mathbf{c}) \subseteq \text{supp}(\hat{\mathbf{c}})$, then $\mathbf{c}\in S_{(d,\infty)}$.
\end{lemma}



We are now in a position to prove Theorem~\ref{thm:rm}.

\begin{proof}[Proof of Theorem~\ref{thm:rm}]
	For a fixed $d\geq 1$, let $z := \left \lceil \log_2(d+1)\right \rceil$. Consider the subcode ${\mathcal{C}_m^{(d)}}$ of the code ${\mathcal{C}_m}$, defined as:
	\begin{align*}
	{\mathcal{C}_m^{(d)}}:=\Bigg\{\text{Eval}(f): f = \bigg(\prod_{i=m-z+1}^{m} x_i \bigg)\cdot h(x_{1},\ldots, x_{m-z}), &\text{ where } \text{deg}(h)\leq r_m-z\Bigg\}.
	\end{align*}
It is easy to verify that the polynomial $q(x_{m-z+1},\ldots,x_m):=\prod_{i=m-z+1}^{m} x_i$ is such that its corresponding evaluation vector, Eval$(q)$, obeys Eval$(q)\in S_{(d,\infty)}^{(2^m)}$. This is because ${q(\mathbf{y})} = 1$ if and only if $(y_{m-z+1},\ldots,y_m) = (1,\ldots,1)$, and in the lexicographic ordering, such evaluation points $\mathbf{y}$ are spaced {apart} by $2^z - 1$ coordinates, where $2^z-1\geq d$. Now, for any polynomial $f$ such that $\text{Eval}(f)\in {\mathcal{C}_m^{(d)}}$, it is true that supp$(\text{Eval}(f))\subseteq \text{supp}(\text{Eval}(q))$. Hence, Eval$(f) \in S_{(d,\infty)}^{(2^m)}$ via Lemma \ref{lem:simplerm}. 

Finally, the rate of the subcode ${\mathcal{C}_m^{(d)}}$ can be calculated as follows:
\begin{align*}
	\text{rate}\left({\mathcal{C}_m^{(d)}}\right) \ & =\ \frac{\log_2\left(\left\lvert{\mathcal{C}_m^{(d)}}\right\rvert\right)}{2^m} \\ 
	&= \ \frac{{m-z \choose \leq r_m-z}}{2^m} 
		\ = \ \frac{\binom{m-z}{\le r_m-z}}{\binom{m-z}{\le r_m}} \, \frac{{m-z \choose \leq r_m}}{2^m}
		\ \xrightarrow{m \to \infty} \ 2^{-z}R.
\end{align*}
 To obtain the limit as $m \to \infty$, we have used Lemma~\ref{lem:rate} and the fact that the ratio $\frac{\binom{m-z}{\le r_m-z}}{\binom{m-z}{\le r_m}}$ converges to $1$ as $m \to \infty$. Towards proving this fact, note that 
 \begin{align*}\binom{m-z}{\le r_m-z} &= \binom{m-z}{\le r_m} - \sum_{i=r_m-z+1}^{r_m}\binom{m-z}{i}\\ &\geq \binom{m-z}{\le r_m} - z\cdot \binom{m-z}{\left\lfloor \frac{m-z}{2}\right\rfloor},\end{align*} and hence, we have that
 \[
1-\frac{z\cdot \binom{m-z}{\left\lfloor \frac{m-z}{2}\right\rfloor}}{\binom{m-z}{\leq r_m}}\ \leq\ \frac{\binom{m-z}{\le r_m-z}}{\binom{m-z}{\le r_m}}\ \leq 1.
 \]
Now, consider the expression $\frac{z\cdot \binom{m-z}{\left\lfloor \frac{m-z}{2}\right\rfloor}}{\binom{m-z}{\leq r_m}}$. We have that $\lim_{m \to \infty}\frac{z\cdot \binom{m-z}{\left\lfloor \frac{m-z}{2}\right\rfloor}}{2^m} = 0$ (in fact, $\frac{z\cdot \binom{m-z}{\left\lfloor \frac{m-z}{2}\right\rfloor}}{2^m} = O\left( \frac{1}{\sqrt{m-z}}\right)$) and that $\lim_{m \to \infty}\frac{\binom{m-z}{\leq r_m}}{2^m} = 2^{-z}R$ (from Lemma \ref{lem:rate}). Hence, it follows that $1-\frac{z\cdot \binom{m-z}{\left\lfloor \frac{m-z}{2}\right\rfloor}}{\binom{m-z}{\leq r_m}}$ converges to $1$, as $m\to \infty$, implying that the ratio $\frac{\binom{m-z}{\le r_m-z}}{\binom{m-z}{\le r_m}}$ also converges to $1$ as $m \to \infty$.
\end{proof}

{\subsection{Existence of Larger (Potentially) Non-Linear $(1,\infty)$-RLL Constrained Subcodes}}
We now proceed to proving Theorem \ref{thm:nonlin}, which establishes the existence of $(1,\infty)$-RLL subcodes of rates better than those in Theorem \ref{thm:rm}. Before we do so, we state and prove a useful lemma on the expected number of runs of $1$s in a codeword of a linear code with dual distance at least $3$. Let $\mathcal{C}^{\perp}$ denote the dual code of a given length-$n$ linear code $\mathcal{C}$, and for a binary vector $\mathbf{v} \in \{0,1\}^n$, let $\tau_0(\mathbf{v})$ and $\tau_1(\mathbf{v})$ be the number of runs of $0$s and $1$s, respectively, in $\mathbf{v}$, with $\tau(\mathbf{v}):= \tau_0(\mathbf{v})+\tau_1(\mathbf{v})$. {Further, given a set $A$, we define the indicator function $\mathds{1}\{x\in A\}$ to be $1$ when $x\in A$, and $0$, otherwise.}
\begin{lemma}
	\label{lem:runcount}
	Let $\mathcal{C}$ be an $[N,K]$ linear code with $\text{d}_{\text{min}}(\mathcal{C}^{\perp})\geq 3$. Then, by drawing codewords $\mathbf{c} \in \mathcal{C}$ uniformly at random, we have that
	\[
	\mathbb{E}_{\mathbf{c}\sim \text{Unif}(\mathcal{C})}[\tau(\mathbf{c})] = \frac{N+1}{2}.
	\]
	Further,
	\[
	\mathbb{E}_{\mathbf{c}\sim \text{Unif}(\mathcal{C})}[\tau_0(\mathbf{c})] = \mathbb{E}_{\mathbf{c}\sim \text{Unif}(\mathcal{C})}[\tau_1(\mathbf{c})] =  \frac{N+1}{4}.
	\]
\end{lemma}
\begin{proof}
	To prove the first part, we note that for any $\mathbf{c}\in \mathcal{C}$, whose coordinates are indexed by $0,1,2,\ldots,N-1$,
	\begin{equation}
		\label{eq:tau}
	\tau(\mathbf{c}) = 1+\#\{0\leq i\leq N-2: (c_i,c_{i+1}) =  (0,1) \text{ or } (c_i,c_{i+1}) =(1,0)\}.
	\end{equation}
	Further, since $\mathcal{C}^{\perp}$ has distance at least $3$, it implies that in any two coordinates $i\neq j$ of $\mathcal{C}$, all binary $2$-tuples occur, and each with frequency $\frac{1}{4}$ {(see e.g.\ \cite[Chapter~5, Theorem~8]{mws} for a proof)}. In particular, from \eqref{eq:tau}, we have that
	\begin{align*}
		\mathbb{E}_{\mathbf{c}\sim \text{Unif}(\mathcal{C})}[\tau(\mathbf{c})] &=1+\sum_{i=0}^{N-2}\mathbb{E}
			[\mathds{1}\{(c_i,c_{i+1}) =  (0,1)\}+\mathds{1}\{(c_i,c_{i+1}) =  (1,0)\}]\\
			&=1+ \sum_{i=0}^{N-2} \bigg(\text{Pr}[(c_i,c_{i+1}) = (0,1)]+ \text{Pr}[(c_i,c_{i+1}) = (1,0)]\bigg)\\
			&= 1+\sum_{i=0}^{N-2}\frac{1}{2}
			= \frac{N+1}{2}.
	\end{align*}
To prove the second part, we note that since $\mathbb{E}[\tau(\mathbf{c})] = \mathbb{E}[\tau_0(\mathbf{c})]+\mathbb{E}[\tau_1(\mathbf{c})]$, it suffices to show that $\mathbb{E}[\tau_0(\mathbf{c})] = \mathbb{E}[\tau_1(\mathbf{c})]$, when $\mathbf{c}$ is drawn uniformly at random from $\mathcal{C}$. To this end, observe that given any codeword $\mathbf{c}\in \mathcal{C}$, we have that
\[
\tau_1(\mathbf{c}) - \tau_0(\mathbf{c}) = 
\begin{cases}
	1,\ \text{if $c_0=c_{N-1} = 1$},\\
	-1,\ \text{if $c_0 = c_{N-1} = 0$},\\
	0,\ \text{otherwise}.
\end{cases}
\]
Hence, when $\mathbf{c}$ is drawn uniformly at random from $\mathcal{C}$,
\begin{align*}
	\mathbb{E}[\tau_1(\mathbf{c})] - \mathbb{E}[\tau_0(\mathbf{c})] &= \mathbb{E}[\mathds{1}\{(c_0,c_{N-1}) = (1,1)\}] - \mathbb{E}[\mathds{1}\{(c_0,c_{N-1}) = (0,0)\}]\\
	&= \frac14 - \frac14 = 0,
\end{align*}
thereby proving that $\mathbb{E}[\tau_0(\mathbf{c})] = \mathbb{E}[\tau_1(\mathbf{c})]$.
\end{proof}
The following standard coding-theoretic fact will also prove useful (see \cite{mws} or \cite{verapless} for discussions on shortening codes):
\begin{lemma}
	\label{lem:shorten}
	Consider an $[N,K]$ linear code $\mathcal{C}$, and let $T\subseteq [0:N-1]$ be a collection of its coordinates. If $\mathcal{C}\big\rvert_T$ denotes the restriction of $\mathcal{C}$ to the coordinates in $T$, and $\mathcal{C}_T$ denotes the code obtained by shortening $\mathcal{C}$ at the coordinates in $T$, then
	\[
	\text{dim}(\mathcal{C}_T) = K - \text{dim}(\mathcal{C}\big\rvert_T).
	\]
	In particular, we have that
	\[
	\text{dim}(\mathcal{C}_T) \geq K - |T|.
	\]
\end{lemma}
Now, we move on to the proof of Theorem \ref{thm:nonlin}.

\begin{proof}[Proof of Theorem \ref{thm:nonlin}]
	Fix any sequence of codes $\left\{{\hat{\mathcal{C}}_m} = \text{RM}(m,r_m)\right\}_{m\geq 1}$ of rate $R\in (0,1)$. Let ${{\mathcal{H}}_m^{(1)}}$ be the largest $(1,\infty)$-RLL subcode of ${\hat{\mathcal{C}}_m}$. Let ${\hat{\mathcal{C}}_{m,+}}:=\text{RM}(m-1,r_m)$ and ${\hat{\mathcal{C}}_{m,-}}:=\text{RM}(m-1,r_m-1)$. The following set of equalities then holds:
	\begin{align}
		\left\lvert{{\mathcal{H}}^{(1)}_m}\right \rvert &= \sum_{f:\ \text{deg}(f)\leq r_m}\mathds{1}\{\text{Eval}(f)\in S_{(1,\infty)}\} \notag\\
		&\stackrel{(a)}{=} \sum_{g: \ \text{deg}(g)\leq r_m} \sum_{h: \ \text{deg}(h)\leq r_m-1} \mathds{1}\{(g,h)\text{ satisfy (C1) and (C2)}\} \notag\\
		&=  \left\lvert{\hat{\mathcal{C}}_{m,+}}\right\rvert \cdot \mathbb{E}_{\text{Eval}(g)\sim \text{Unif}({\hat{\mathcal{C}}_{m,+}})} \left[ \sum_{h: \ \text{deg}(h)\leq r_m-1} \mathds{1}\{(g,h)\text{ satisfy (C1) and (C2)}\}\right] \notag\\
		&=\left\lvert{\hat{\mathcal{C}}_{m,+}}\right\rvert \cdot \mathbb{E}_{\text{Eval}(g)\sim \text{Unif}({\hat{\mathcal{C}}_{m,+}})} \left[\#\{h:(g,h)\text{ satisfy (C1) and (C2)} \}\right], \label{eq:temp1}
	\end{align}
where (a) holds from Proposition \ref{lem:necesscond}. The following fact will be of use to us: if ${\hat{R}_{m,+}}:=\text{rate}\left({\hat{\mathcal{C}}_{m,+}}\right)$ and ${\hat{R}_{m,-}}:=\text{rate}\left({\hat{\mathcal{C}}_{m,-}}\right)$, then
\begin{align*}
	 {\hat{R}_{m,+}}+{\hat{R}_{m,-}}&= \frac{{m-1\choose \leq r_m}+{m-1\choose \leq r_m-1}}{2^{m-1}}  
	\xrightarrow{m\to \infty} 2R. 
\end{align*}
Now, from the definitions of (C1) and (C2) in Proposition \ref{lem:necesscond}, we have that for a fixed Eval$(g) \in {\hat{\mathcal{C}}_{m,+}}$,
\begin{align*}
\#\{h:(g,h)\text{ satisfy (C1) and (C2)} \} &=\#\{h: h(\mathbf{b}) = 1\ \forall\ \mathbf{b}\in \text{supp}(\text{Eval}(g)),\text{ and }h(\mathbf{b}^\prime) = 0\ \forall\ \mathbf{b}^\prime\in \Gamma(g)\}
\end{align*}
The right-hand side of the above equality is precisely the number of codewords in the code obtained by shortening ${\hat{\mathcal{C}}_{m,-}}$ at the coordinates in $S:= \text{supp}(\text{Eval}(g)) \cup \Gamma(g)$. Note that $$|S| = \left\lvert \text{supp}(\text{Eval}(g))\right\rvert + \left\lvert \Gamma(g) \right \rvert \leq \text{wt}(\text{Eval}(g))+\tau_0(\text{Eval}(g)).$$
Hence, from the second part of Lemma \ref{lem:shorten}, we have that
\begin{align}
	\label{eq:shorten}
\#\{h:(g,h)\text{ satisfy (C1) and (C2)} \} \geq \text{exp}_2\left({2^{m-1}\cdot {\hat{R}_{m,-}}-\left(\text{wt}(\text{Eval}(g))+\tau_0(\text{Eval}(g))\right)}\right)
\end{align}
Plugging \eqref{eq:shorten} back in equation \eqref{eq:temp1}, we get that for $m$ large enough,
\begin{align}
	\left\lvert{\mathcal{H}_m^{(1)}}\right \rvert &\geq \left\lvert{\hat{\mathcal{C}}_{m,+}}\right\rvert \cdot \mathbb{E}_{\text{Eval}(g)\sim \text{Unif}({\hat{\mathcal{C}}_{m,+}})}  \left[\text{exp}_2\left({2^{m-1}\cdot\left({\hat{R}_{m,-}}-\frac{\text{wt}(\text{Eval}(g))}{2^{m-1}}-\frac{\tau_0(\text{Eval}(g))}{2^{m-1}}\right)}\right)\right] \label{eq:nonlinimprov}\\
	&\stackrel{(b)}{\geq} \left\lvert{\hat{\mathcal{C}}_{m,+}}\right\rvert \cdot \text{exp}_2\left({2^{m-1}\cdot\left({\hat{R}_{m,-}} - \frac{\mathbb{E}[\text{wt}(\text{Eval}(g))]}{2^{m-1}}-\frac{\mathbb{E}[\tau_0(\text{Eval}(g))]}{2^{m-1}}\right)}\right) \notag\\
	&\stackrel{(c)}{=} \left\lvert{\hat{\mathcal{C}}_{m,+}}\right\rvert \cdot \text{exp}_2\left({2^{m-1}\cdot\left({\hat{R}_{m,-}} - \frac12-\frac14 -\delta_m\right)}\right) \notag\\
	&\stackrel{(d)}{=} \text{exp}_2\left(2^{m}\cdot \left( \frac{{\hat{R}_{m,+}}+{\hat{R}_{m,-}}}{2}-\frac38 -\frac{\delta_m}{2}\right)\right), \notag
\end{align}
where $\delta_m:=\frac{1}{4\cdot2^{m-1}} \xrightarrow{m\to \infty} 0$. Here, (b) holds by an application of Jensen's inequality and the linearity of expectation. {To see why (c) holds, we note that the RM code $\hat{\mathcal{C}}_{m,+}$ has no coordinate that is identically $0$. Thus, in a randomly chosen codeword $\text{Eval}(g)\sim \text{Unif}(\hat{\mathcal{C}}_{m,+})$, every coordinate is equally likely to be $0$ or $1$, and hence, 
$$\mathbb{E}[\text{wt}(\text{Eval}(g))] = \sum_{\mathbf{z}\in \{0,1\}^{m-1}}\Pr[g(\mathbf{z}) = 1] = 2^{m-2}.$$
Moreover, by Lemma \ref{lem:runcount}, we have that
$$\mathbb{E}[\tau_0(\text{Eval}(g))] = \frac{2^{m-1}+1}{4}.$$}
Finally, (d) holds from the fact that $\left\lvert{\hat{\mathcal{C}}_{m,+}}\right\rvert = \text{exp}_2\left(2^{m-1}\cdot {\hat{R}_{m,+}}\right)$.

Hence, we get that the largest rate of $(1,\infty)$-RLL constrained subcodes of $\left\{{\hat{\mathcal{C}}_m}\right\}_{m\geq 1}$ obeys
\begin{align*}
	{\mathsf{R}^{(1)}(\hat{\mathcal{C}})} &= \limsup_{m\to \infty} \max_{{\mathcal{H}_m^{(1)}}\subseteq {\hat{\mathcal{C}}_m}}\frac{\log_2\left\lvert{\mathcal{H}_m^{(1)}}\right\rvert}{2^m}\\
	&\geq \limsup_{m\to \infty}\frac{2^{m}\cdot \left(\frac{{\hat{R}_{m,+}}+{\hat{R}_{m,-}}}{2}-\frac38-\frac{\delta_m}{2}\right)}{2^m}\\
	&= R-\frac38.
\end{align*}
Thus, there exists a sequence of $(1,\infty)$-RLL constrained subcodes of any sequence of RM codes $\{{\hat{\mathcal{C}}_m}\}_{m\geq 1}$ of rate $R$, such that the subcodes are of rate of at least $\max\left(0,R-\frac38\right)$.
\end{proof}

Although Theorem \ref{thm:nonlin} proves the existence of non-linear $(1,\infty)$-RLL subcodes of rate larger than (for high rates $R$) that of the linear subcodes of Theorem \ref{thm:rm}, it is of interest to check if further improvements on the rates of $(1,\infty)$-RLL constrained subcodes are possible, by performing numerical computations. For a fixed $R\in (0,1)$, we work with the sequence of RM codes $\left\{{\hat{\mathcal{C}}_m} = \text{RM}(m,v_m)\right\}_{m\geq 1}$, where 
\[
v_m:= \min \left\{u: \sum_{i=0}^{u} \binom{m}{i} \geq \left\lfloor2^m\cdot R\right\rfloor\right\}.
\]
It can be checked that rate$\left({\hat{\mathcal{C}}_m}\right)\xrightarrow{m\to \infty}R$. 
We then have from inequality \eqref{eq:nonlinimprov} that
\begin{equation}
	\label{eq:nonlinlogexp}
{\mathsf{R}^{(1)}(\hat{\mathcal{C}})} \geq \max \left\{0,\limsup_{m\to \infty} \left( \frac{{\hat{R}_{m,+}}+{\hat{R}_{m,-}}}{2}+\frac{\log_2\left(\mathbb{E}\left[\text{exp}_2\left(-\text{wt}(\text{Eval}(g))-\tau_0(\text{Eval}(g))\right)\right]\right)}{2^m}\right)\right\},
\end{equation}
where the expectation is taken over codewords $\text{Eval}(g)\sim \text{Unif}({\hat{\mathcal{C}}_{m,+}})$. Inequality \eqref{eq:nonlinlogexp} suggests that one can estimate a lower bound on ${\mathsf{R}^{(1)}(\hat{\mathcal{C}})}$, by picking a large $m$ and replacing the expectation by a sample average over codewords $\text{Eval}(g)$ chosen uniformly at random from ${\hat{\mathcal{C}}_{m,+}}$. We can then obtain a new (numerical) lower bound, which does not make use of a further lower bounding argument via Jensen's inequality. We performed this Monte-Carlo sampling and estimation procedure, with $m=11$, and for varying values of $R$, by averaging over $10^4$ uniformly random samples of codewords, $\text{Eval}(g)$. Figure \ref{fig:lbnumerical} shows a plot comparing the lower bound in \eqref{eq:nonlinlogexp}, with the lower bound that is approximately $\hat{R}_m -\frac38$, from Theorem \ref{thm:nonlin}, where $\hat{R}_m := $ rate$\left({\hat{\mathcal{C}}_m}\right)$. We observe that there is a noticeable improvement in the numerical rate lower bound, as compared to the bound in Theorem \ref{thm:nonlin}, for some values of $\hat{R}_m$.

\begin{figure*}%
	\centering

	\includegraphics[width=0.8\textwidth]{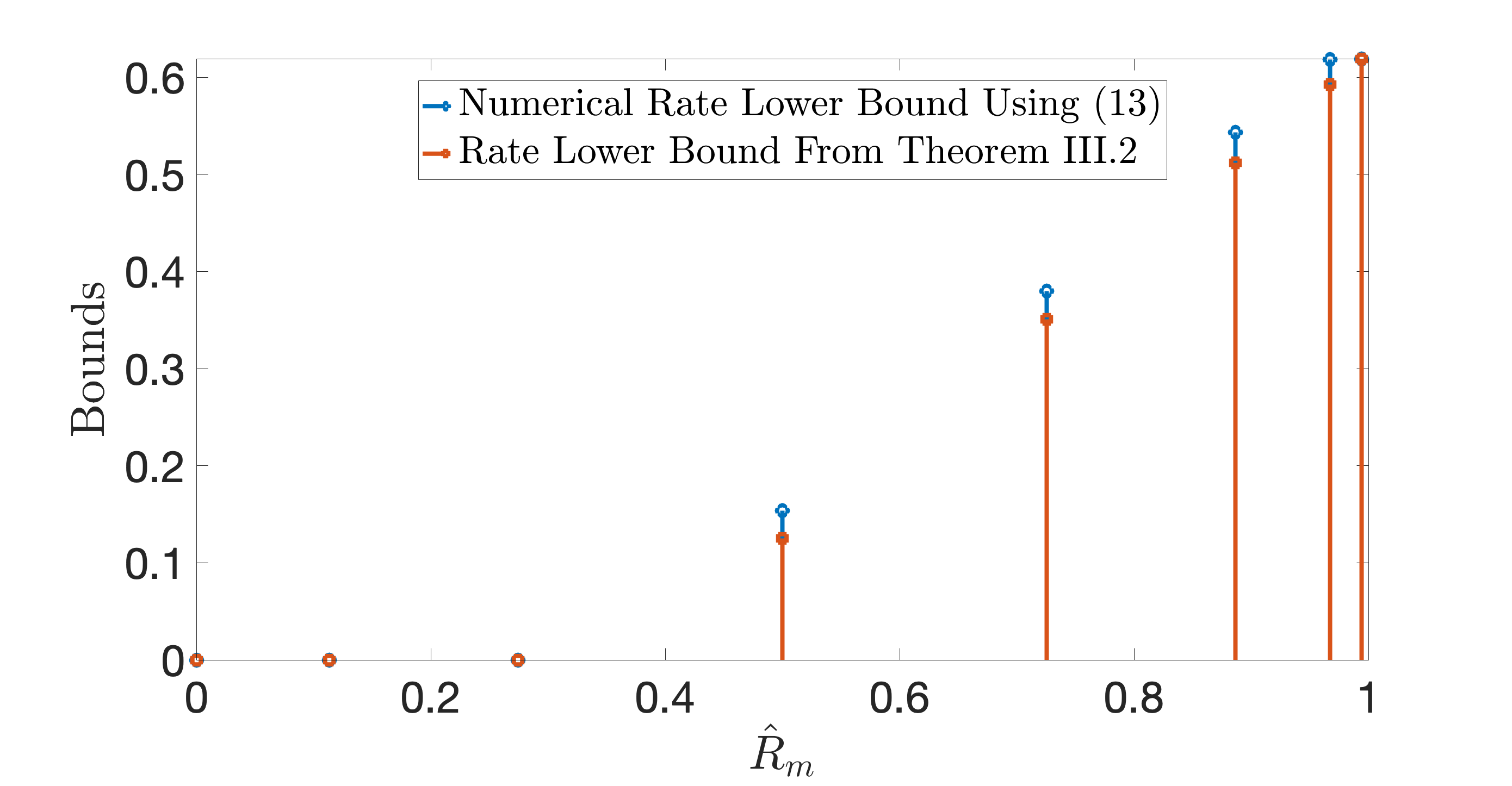}%
	\caption{Plot comparing, for $d=1$, the rate lower bound of approximately $\hat{R}_m - \frac38$, from Theorem \ref{thm:nonlin}, with the numerical lower bound obtained by Monte-Carlo simulation, using \eqref{eq:nonlinlogexp}.}
	\label{fig:lbnumerical}%
\end{figure*}
\section{Upper Bounds on Rates of Constrained Subcodes}
\label{sec:ub}
In this section, we derive upper bounds on the rates of $(d,\infty)$-RLL constrained subcodes of RM codes of rate $R\in (0,1)$. In the first subsection, we restrict our subcodes to be linear, while in the second subsection, we fix $d$ to be $1$ and {drop} the assumption of linearity of the subcodes. Additionally, in the first two subsections, we assume that the RM codes follow a lexicographic coordinate ordering. We consider RM codes under other coordinate orderings in the last subsection, and derive upper bounds on the rates of linear $(d,\infty)$-RLL constrained subcodes.
{\subsection{Linear $(d,\infty)$-RLL Constrained Subcodes Under Lexicographic Coordinate Ordering}}
\label{sec:rmlinub}
We fix a sequence of codes $\left\{{\hat{\mathcal{C}}_m} = \text{RM}(m,r_m)\right\}_{m\geq 1}$ of rate $R$, whose coordinates follow a lexicographic ordering. We first state and prove a fairly general proposition on the rates of linear $(d,\infty)$-RLL subcodes of linear codes. Recall that for an $[N,K]$ linear code $\mathcal{C}$ over $\mathbb{F}_2$, of blocklength $N$ and dimension $K$, an inrformation set is a collection of $K$ coordinates in which all possible $K$-tuples over $\mathbb{F}_2$ can appear. Equivalently, if $G$ is any generator matrix for $\mathcal{C}$, an information set is a set of $K$ column indices such that $G$ restricted to those columns is a full-rank matrix. As usual, we index the coordinates of the code from $0$ to $N-1$.

\begin{proposition}
	\label{prop:linub}
	Let ${\mathcal{C}}$ be an $[N,K]$ binary linear code. If $\mathcal{I}$ is an information set of ${\mathcal{C}}$ that contains $t$ disjoint $(d+1)$-tuples of consecutive coordinates $(i_1, i_1 + 1,\ldots,i_1+d), (i_2, i_2 + 1,\ldots,i_2+d), ..., (i_t, i_t + 1,\ldots,i_t+d)$, with $i_1\geq 0$, $i_j>i_{j-1}+d$, for all $j\in [2:t]$, and $i_t\leq N-1-d$, then the dimension of any linear $(d,\infty)$-RLL subcode of ${\mathcal{C}}$ is at most $K-dt$. 
\end{proposition}
\begin{proof}
	Suppose that the information set $\mathcal{I}$ contains exactly $t$ disjoint $(d+1)$-tuples of consecutive coordinates as in the statement of the proposition. By definition, all possible $K$-tuples appear in the coordinates in $\mathcal{I}$. Now, consider any linear $(d,\infty)$-RLL subcode $\overline{\mathcal{C}}$ of ${\mathcal{C}}$, and any $(d+1)$-tuple of consecutive coordinates $\{i_j,i_j+1,\ldots,i_j+d\} \in \mathcal{I}$, for $j\in [t]$. Since the $(d,\infty)$-RLL constraint requires that successive $1$s be separated by at least $d$ $0$s, the only possible tuples of $d+1$ consecutive symbols in any codeword of the linear subcode are $(0,0,\ldots,0)$ and at most one of $\mathbf{e}_i^{(d+1)}$, $i\in [d+1]$. This is because, if $\mathbf{e}_i^{(d+1)}$ and $\mathbf{e}_j^{(d+1)}$ both occur in a collection of $d+1$ consecutive positions, then, by linearity of the subcode $\overline{\mathcal{C}}$, we have $\mathbf{e}_i^{(d+1)}+\mathbf{e}_j^{(d+1)}$ (where the addition is over vectors in $\mathbb{F}_2^{d+1}$) must occur in some codeword of the subcode, thereby making the codeword not $(d,\infty)$-RLL compliant. Hence, for every $(d+1)$-tuple of consecutive coordinates in $\mathcal{I}$, only a $2^{-d}$ fraction of the $2^{d+1}$ possible tuples are allowed. Thus, overall, the number of possible $K$-tuples that can appear in the information set $\mathcal{I}$ in the codewords of the linear subcode is at most $\frac{2^K}{2^{dt}}$. Hence, the number of codewords in the linear subcode is at most $2^{K-dt}$.
	
\end{proof}
In order to obtain an upper bound, as in Theorem \ref{thm:rmlinub}, on the rate of linear $(d,\infty)$-RLL subcodes of the sequence of codes $\left\{{\hat{\mathcal{C}}_m}\right\}_{m\geq 1}$, we shall first identify an information set  of ${{\hat{\mathcal{C}}_m}} = \text{RM}(m,r_m)$. We then compute the number of disjoint $(d+1)$-tuples of consecutive coordinates in the information set, and apply Proposition \ref{prop:linub} to get an upper bound on the dimension of the linear constrained subcodes.

Given the integers $m$ and $r$, consider the binary linear code $\tilde{\mathcal{C}}(m,r)$ (which is a subspace of $\mathbb{F}_2^{2^m}$), spanned by the codewords in the set 
\begin{equation}
	\label{eq:Bmr}
	\mathcal{B}_{m,r}:=\left\{\text{Eval}\left(\prod_{i\in S}x_i\right): S\subseteq [m]\ \text{with } |S|\geq r+1\right\}.
\end{equation}

From the discussion in Section \ref{sec:rmintro}, we observe that the vectors in $\mathcal{B}_{m,r}$ are linearly independent, and, hence, $\mathcal{B}_{m,r}$ forms a basis for $\tilde{\mathcal{C}}(m,r)$, with dim$\left(\tilde{\mathcal{C}}(m,r)\right) = {m \choose \geq r+1}$. Moreover, the codewords in $\tilde{\mathcal{C}}(m,r)$ are linearly independent from codewords in RM$(m,r)$, since the evaluation vectors of all the distinct monomials in the variables $x_1,\ldots,x_m$ are linearly independent over $\mathbb{F}_2$.

The following lemma identifies an alternative basis for $\tilde{\mathcal{C}}(m,r)$, which will prove useful in our analysis, later on.

\begin{lemma}
	\label{lem:quotient}
	Consider the code $\tilde{\mathcal{C}}(m,r)=\text{span}\left(\mathcal{B}_{m,r}\right)$, where $\mathcal{B}_{m,r}$ is as in \eqref{eq:Bmr}. We then have that $\tilde{\mathcal{C}}(m,r) = \text{span}\left(\{\mathbf{e}_{\mathbf{b}}: \text{wt}(\mathbf{b})\geq r+1\}\right)$.
	\end{lemma}
\begin{proof}
	
	Note that any standard basis vector $\mathbf{e}_{\mathbf{b}}$, with  wt$(\mathbf{b})\geq r+1$, can be written as Eval$(f)$, where 
	\[
	f(x_1,\ldots,x_m) = \prod_{i\in \text{supp}(\mathbf{b})}x_i \cdot \prod_{i\notin \text{supp}(\mathbf{b})}(1+x_j).
	\]
	From the fact that wt$(\mathbf{b})\geq r+1$, we have that the degree of any monomial in $f$ is at least $r+1$, and hence, Eval$(f) = \mathbf{e}_\mathbf{b} \in \text{span}(\mathcal{B}_{m,r}) = \tilde{\mathcal{C}}(m,r)$. The result follows by noting that $\{\mathbf{e}_{\mathbf{b}}: \text{wt}(\mathbf{b})\geq r+1\}$ is a collection of linearly independent vectors of size ${m\choose \geq r+1}$, which, in turn, equals dim$\left(\tilde{\mathcal{C}}(m,r)\right)$.
\end{proof}
We now introduce some further notation: given a $p\times q$ matrix $M$, we use the notation $M[\mathcal{U},\mathcal{V}]$ to denote the submatrix of $M$ consisting of the rows in the set $\mathcal{U}\subseteq [p]$ and the columns in the set $\mathcal{V}\subseteq [q]$. We recall the definition of the generator matrix $G_{\text{Lex}}(m,r)$, of RM$(m,r)$, and the indexing of columns of the matrix, from Section \ref{sec:rmintro}. Finally, towards identifying an information set of RM$(m,r)$, we define the set of coordinates
\begin{equation}
\label{eq:Imr}
\mathcal{I}_{m,r}:= \{\mathbf{b} = (b_1,\ldots,b_m)\in \mathbb{F}_2^m: \text{wt}(\mathbf{b})\leq r\}.
\end{equation}
\begin{lemma}
	\label{lem:infoset}
	The set of coordinates $\mathcal{I}_{m,r}$ is an information set of $\text{RM}(m,r)$.
\end{lemma}
\begin{proof}
	 In order to prove that $\mathcal{I}_{m,r}$ is an information set of RM$(m,r)$, it is sufficient to show that $G_{\text{Lex}}(m,r)$ restricted to the columns in $\mathcal{I}_{m,r}$ is full rank.
	
	Now, consider the generator matrix $\tilde{G}(m,r)$, of $\tilde{\mathcal{C}}(m,r)$, consisting of rows that are vectors in $\mathcal{B}_{m,r}$. We build the $2^m\times 2^m$ matrix $$\mathsf{H}:= \begin{bmatrix} \begin{array}{c}
		\tilde{G}(m,r)\\
		\hline\\
		G_{\text{Lex}}(m,r)
		\end{array}
	\end{bmatrix},$$
with $\mathsf{H}$ being full rank. Note that, from Lemma \ref{lem:quotient}, any standard basis vector $\mathbf{e}_{\mathbf{b}}$, with $\mathbf{b} \in \mathcal{I}_{m,r}^c$, belongs to rowspace$(\tilde{G}(m,r))$. By Gaussian elimination, it is then possible to replace the first ${m\choose \geq r+1}$ rows of $\mathsf{H}$, corresponding to the submatrix $\tilde{G}(m,r)$, with the standard basis vectors $\mathbf{e}_\mathbf{b}$, with $\mathbf{b}\in \mathcal{I}_{m,r}^c$. Clearly, from the fact that $\mathsf{H}$ is full rank, this then means that $\mathsf{H}\left[\left[{m\choose \geq r+1}+1:2^m\right],\mathcal{I}_{m,r}\right]$ is full rank, or equivalently, $G_{\text{Lex}}(m,r)$ restricted to columns in $\mathcal{I}_{m,r}$ is full rank.
\end{proof}

We thus have that an information set of ${\hat{\mathcal{C}}_m}$ is $\mathcal{I}_{m,r_m}$. Note that in order to get an upper bound on the rate of linear $(d,\infty)$-RLL subcodes of ${\hat{\mathcal{C}}_m}$, we need only calculate the number of disjoint $(d+1)$-tuples of consecutive coordinates in $\mathcal{I}_{m,r_m}$.  


To this end, for any $0\leq r\leq m-1$, we define
\begin{align}
	\label{eq:gammam}
\Gamma_{m,r} := \left\{s\in [0:2^m-2]: \mathbf{B}(s)\in \mathcal{I}_{m,r},\text{ but } \mathbf{B}(s+1)\notin \mathcal{I}_{m,r}\right\},
\end{align}
to be the set of right end-point coordinates of runs that belong to $\mathcal{I}_{m,r}$. The number of such runs is $\left \lvert \Gamma_{m,r}\right \rvert$. Observe that since $r\leq m-1$, we have $2^{m}-1\notin \mathcal{I}_{m,r}$.  In the case where $r=m$, we have that $\mathcal{I}_{m,r} = [0:2^m-1]$, and we define $\Gamma_{m,r}$ to be $\{2^m-1\}$. However, this special case need not be considered, for our purposes.
\begin{lemma}
	\label{lem:runs}
For $0\leq r\leq m-1$, the equality $\left \lvert \Gamma_{m,r}\right \rvert={m-1\choose r}$ holds.
\end{lemma}
\begin{proof}
	Let $r\in [0:m-1]$. Note that every right end-point of a run, $s\in \Gamma_{m,r}$,  with $s\in [0:2^m-2]$, is such that wt$(\mathbf{B}(s))\leq r$, but wt$(\mathbf{B}(s+1))\geq r+1$. We now claim that an integer $s \in \Gamma_{m,r}$ iff \textbf{B}$(s)=(b_1,\ldots ,b_{m-1},0)$, for $b_1,\ldots,b_{m-1}\in \{0,1\}$, with wt$((b_1,\ldots ,b_{m-1},0)) = r$. 
	
	To see this, note that if \textbf{B}$(s)=(b_1,\ldots, b_{m-1},0)$, then \textbf{B}$(s+1) = (b_1,\ldots, b_{m-1},1)$. Hence, if wt$((b_1,\ldots, b_{m-1})) = r$, then $s\in \Gamma_{m,r}$. Conversely, if $s\in \Gamma_{m,r}$, then \textbf{B}$(s)$ cannot end in a $1$. Indeed, if this were the case, then we would have \textbf{B}$(s)$ being of the form $(b_1,\ldots, b_\ell, 0,1, \ldots, 1)$, with $b_1,\ldots,b_{\ell}\in \{0,1\}$, so that \textbf{B}$(s+1)$ would be $(b_1,\ldots, b_\ell, 1,0, \ldots, 0)$, the weight of which does not exceed that of \textbf{B}$(s)$. So, \textbf{B}$(s)$ must be of the form $(b_1, \ldots, b_{m-1}, 0)$, and so, \textbf{B}$(s+1) = (b_1,\ldots, b_{m-1},1)$. From wt$(\mathbf{B}(s)) \le r$ and wt$(\mathbf{B}(s+1)) \ge r+1$, we obtain that wt$(b_1 \ldots b_{m-1}) = r$. 
	
	This then implies that the number of runs, which is equal to the number of right end-points of runs, exactly equals ${m-1\choose r}$.
\end{proof}
With the ingredients in place, we are now in a position to prove Theorem \ref{thm:rmlinub}.
\begin{proof}[Proof of Theorem \ref{thm:rmlinub}]
	Fix a sequence of codes $\left\{{\hat{\mathcal{C}}_m} = \text{RM}(m,r_m)\right\}_{m\geq 1}$ of rate $R\in (0,1)$, with $r_m\leq m-1$, for all $m$. We use the notation $K_m:={m\choose \leq r_m}$ to denote the dimension of ${{\hat{\mathcal{C}}_m}}$. 
		
	Now, for a given $m$, consider the information set $\mathcal{I}_{m,r_m}$ (see \eqref{eq:Imr}). We know from Lemma \ref{lem:runs} that the number of runs, $\left \lvert \Gamma_{m,r_m}\right \rvert$, of coordinates that lie in $\mathcal{I}_{m,r_m}$, is exactly ${m-1\choose r_m}$. Now, note that the $i^{\text{th}}$ run $(s_i,\ldots,s_i+\ell_i-1)$, of length $\ell_i$, with $s_i\in \Gamma_{m,r_m}$ and $i\in \left[\left \lvert \Gamma_{m,r_m}\right \rvert\right]$, contributes $\left \lfloor \frac{\ell_i}{d+1}\right \rfloor$ disjoint $(d+1)$-tuples of consecutive coordinates in $\mathcal{I}_{m,r}$. It then holds that the overall number of disjoint $(d+1)$-tuples of consecutive coordinates in $\mathcal{I}_{m,r}$ is $t_m$, where
	\begin{align*}
		t_m &= \sum_{i=1}^{\left \lvert \Gamma_{m,r_m}\right \rvert}\left \lfloor \frac{\ell_i}{d+1} \right \rfloor\\
		&\geq \sum_{i=1}^{\left \lvert \Gamma_{m,r_m}\right \rvert} \left(\frac{\ell_i}{d+1} -1\right)\\
		&= \frac{K_m}{d+1} - \left \lvert \Gamma_{m,r_m} \right \rvert
		= \frac{K_m}{d+1} - {m-1 \choose r_m},
	\end{align*}
where the last equality follows from Lemma \ref{lem:runs}.
	
	Using Proposition \ref{prop:linub}, it follows that the dimension of any linear $(d,\infty)$-RLL subcode ${\overline{\mathcal{H}}^{(d)}} \subseteq {\hat{\mathcal{C}}_m}$ is at most $K_m-dt_m$. It then holds that
	\begin{align*}
			{\overline{\mathsf{R}}^{(d)}(\hat{\mathcal{C}})}&=\limsup_{m\to \infty} \max_{{\overline{\mathcal{H}}^{(d)}}\subseteq {\hat{\mathcal{C}}_m}} \frac{\log_2\left \lvert {\overline{\mathcal{H}}^{(d)}}\right\rvert}{2^m}\\
		&\leq \limsup_{m\to \infty}\frac{K_m-dt_m}{2^m}\\
		&\leq \limsup_{m\to \infty}\frac{K_m-\frac{dK_m}{d+1}+d\cdot {m-1\choose r_m}}{2^m}\\
		&\leq \lim_{m\to \infty}\frac{\frac{K_m}{d+1}+d\cdot {m-1\choose \left \lfloor \frac{m-1}{2}\right \rfloor}}{2^m}\\
		&= \frac{R}{d+1},
	\end{align*}
	where the last equality holds from the fact that ${m-1\choose \left \lfloor \frac{m-1}{2}\right \rfloor}$ is $O\left(\frac{2^m}{\sqrt{m-1}}\right)$, and $\lim_{m\to \infty} \frac{K_m}{2^m} = R$.
\end{proof}

{\subsection{General $(1,\infty)$-RLL Constrained Subcodes Under Lexicographic Coordinate Ordering}}
\label{sec:rmub}
In this section, we provide upper bounds on the rates of $(1,\infty)$-RLL subcodes of $\{{\mathcal{C}_m} = \text{RM}(m,r_m)\}_{m\geq 1}$, where
\[
r_m = \max \left\{\left \lfloor \frac{m}{2}+\frac{\sqrt{m}}{2}Q^{-1}(1-R)\right \rfloor,0\right\}.
\]

Recall that  $\{{\mathcal{C}_m} = \text{RM}(m,r_m)\}_{m\geq 1}$ is a sequence of RM codes of rate $R$. We fix the coordinate ordering to be the standard lexicographic ordering. 

Recall also, from Lemma \ref{lem:necesscond}, that any $(1,\infty)$-RLL subcode of ${\mathcal{C}_m}$ must be such that both conditions (C1) and (C2) are simultaneously satisfied. Therefore, to obtain an upper bound on the number of codewords Eval$(f) \in {{\mathcal{C}_m}}$ that respect the $(1,\infty)$-RLL constraint, it is sufficient to obtain an upper bound on the number of Eval$(f) \in {{\mathcal{C}_m}}$ that satisfy (C1) alone. In other words, we wish to obtain an upper bound on the number of pairs of polynomials $(g,h)$ (see the Plotkin decomposition in \eqref{eq:rmdecomp}), with Eval$(g) \in {\mathcal{C}_{m,+}}$ and Eval$(h)\in {\mathcal{C}_{m,-}}$, such that supp$(\text{Eval}(g)) \subseteq \text{supp}(\text{Eval}(h))$.

The following two lemmas from the literature will be useful in the proof of Theorem \ref{thm:rmub}.
\begin{lemma}[\cite{abbe1}, Lemma 36]
	\label{lem:short}
	Let $\mathcal{V} \subseteq \mathbb{F}_2^m$ be such that $|\mathcal{V}|\geq 2^{m-u}$, for some $u\in \mathbb{N}$. Then,
	\[
	\text{rank}({G_{m,r}}[\mathcal{V}])> {m-u \choose \leq r},
	\]
	where ${G_{m,r}}$ is a generator matrix of RM$(m,r)$, and ${G_{m,r}}[\mathcal{V}]$ denotes the set of columns of ${G_{m,r}}$, indexed by $\mathcal{V}$.
\end{lemma}
The lemma below follows from Theorem 1 in \cite{anuprao}, by noting that the Reed-Muller code is transitive:
\begin{lemma}
	\label{lem:anup}
	{Let the weight distribution of RM$(m,r)$ be $\left(A_{m,r}(w): 0\leq w\leq 2^m\right)$.} Then,
	\[
	A_{m,r}(w)\leq \text{exp}_2\left({m\choose \leq r}\cdot h_b\left(\frac{w}{2^m}\right)\right),
	\]
	where $h_b(\cdot)$ is the binary entropy function.
\end{lemma}
We now provide the proof of Theorem \ref{thm:rmub}.
\begin{proof}[Proof of Theorem~\ref{thm:rmub}]
	Fix the sequence $\{{{\mathcal{C}_m}} = \text{RM}(m,r_m)\}_{m\geq 1}$ of RM codes, with
	\[
		r_m = \max \left\{\left \lfloor \frac{m}{2}+\frac{\sqrt{m}}{2}Q^{-1}(1-R)\right \rfloor,0\right\},
	\]
 as in the statement of the theorem.
	Now, for any codeword Eval$(g) \in {\mathcal{C}_{m,+}}$ of weight $w$, we shall first calculate the number, $N_w(g)$, of codewords Eval$(h) \in {\mathcal{C}_{m,+}}$ such that supp$(\text{Eval}(g)) \subseteq \text{supp}(\text{Eval}(h))$. Note that $N_w(g)$ serves as an upper bound on the number of codewords Eval$(h) \in {\mathcal{C}_{m,-}}$ such that the same property holds.
	
	To this end, suppose that for any weight $w$, the integer  $u=u(w)$ is the smallest integer such that wt$(\text{Eval}(g)) = w \geq2^{m-1-u}$. Note that for any polynomial $g$ as above of weight $w$, the number of codewords in the code produced by shortening ${\mathcal{C}_{m,+}}$ at the indices in supp$(\text{Eval}(g))$, equals $N_w(g)$.  Now, since dim$\left({\mathcal{C}_{m,+}}\right) = {m-1\choose \leq r_m}$, and from the fact that $w\geq 2^{m-1-u}$, we obtain by an application of Lemmas \ref{lem:shorten} and \ref{lem:short}, that
	\begin{align}
		\label{eq:nw}
	N_w(g)&\leq \text{exp}_2\left({m-1\choose \leq r_m} - {m-1-u \choose \leq r_m}\right)\\
	&=:M_{u(w)}. \notag
	\end{align}
	
	Let ${\mathcal{H}^{(1)}_m}$ denote the largest $(1,\infty)$-RLL subcode of ${{\mathcal{C}_m}}$, and let  ${\mathsf{R}_{m}^{(1)}(\mathcal{C})}$ denote its rate, i.e.,
	\[
	{\mathsf{R}_{m}^{(1)}(\mathcal{C})} = \frac{\log_2\left(\left\lvert{\mathcal{H}^{(1)}_m} \right\rvert\right)}{2^m}.
	\]
	Then, 
	\begin{align}
\left\lvert{\mathcal{H}^{(1)}_m} \right\rvert&\leq \sum_{g: \text{Eval}(g)\in {\mathcal{C}_{m,+}}}{N_w(g)} \notag	\\
		&\stackrel{(a)}{\leq} \sum\limits_{w=2^{m-1-r_m}}^{2^{m-1}}A_{m-1,r_m}(w){M_{u(w)}} \notag\\
		&\stackrel{(b)}{\leq} \left\{\sum\limits_{w=2^{m-1-r_m}}^{2^{m-2}}A(w){M_{u(w)}}\right\}+\frac{1}{2}\cdot \text{exp}_2\left(m-1\choose \leq r_m\right)\cdot
		\text{exp}_2\left({m-1\choose \leq r_m} - {m-2\choose \leq r_{m}}\right)\notag \\
		&\stackrel{(c)}{\leq} \left\{\sum\limits_{w=2^{m-1-r_m}}^{2^{m-2}}A(w){M_{u(w)}}\right\}+
		\frac{1}{2}\cdot \text{exp}_2\left({m-1\choose \leq r_m}+{m-2\choose \leq r_{m}}\right)\notag\\
		&\stackrel{(d)}{\leq} \Bigg\{\sum\limits_{i=1}^{{r_m-1}}A\left(\left[2^{m-2-i}: 2^{m-1-i}\right]\right)\cdot \text{exp}_2\left({m-1\choose \leq r_m} - {m-2-i \choose \leq r_m}\right)\Bigg\}+  \frac{1}{2}\cdot\text{exp}_2\left({m-1\choose \leq r_m}+{m-2\choose \leq r_{m}}\right) \label{eq:sub1},
	\end{align}
where, for ease of reading, we write $A(w):=A_{m-1,r_m}(w)$, in inequalities (b)--(d). Further, $A([a:b])$ is shorthand for $\sum_{{w}=a}^{b}A(w)$. Here,
\begin{enumerate}[label = (\alph*)]
	\item follows from \eqref{eq:nw}, and
	\item holds due to the following fact: since the all-ones codeword $\mathbf{1}$ is present in ${\mathcal{C}_{m,+}} = \text{RM}(m-1,r_m)$, it implies that $A(w) = A(2^{m-1}-w)$, i.e., that the weight distribution of ${\mathcal{C}_{m,+}}$ is symmetric about {the} weight $w=2^{m-2}$. Therefore,
	\begin{equation}
		\label{eq:Aw}
		A\left(\left[2^{m-2}+1:2^{m-1}\right]\right)\leq \frac{1}{2}\cdot \text{exp}_2{m-1\choose \leq r_m}.
	\end{equation}%
\end{enumerate}
Next,
\begin{enumerate}
	\item[(c)] follows from the fact that for positive integers $n,k$ with $n>k$:
	\[
	{n-1\choose k}+{n-1\choose k-1} = {n\choose k}.
	\]
	Picking $n=m-1$, we obtain that for $k\leq m-2$,
	\[
	{m-1\choose k}-{m-2\choose k} = {m-2\choose k-1},
	\]
	and hence that
	\[
	{m-1\choose \leq r_m}-{m-2\choose \leq r_m} < {m-2\choose \leq r_m}, \ \text{and}
	\]
	
	\item[(d)] holds again by \eqref{eq:nw}.
\end{enumerate}
It is clear that a simplification of \eqref{eq:sub1} depends crucially on good upper bounds on the weight distribution function. 
Recall the notation ${R_{m,+}}:= \frac{{m-1\choose \leq r_m}}{2^{m-1}}$. We now use the result in Lemma \ref{lem:anup}, to get that for $1\leq i\leq r_m-1$,
\begin{align}
	A_{m-1,r_m}\left(\left[2^{m-2-i}: 2^{m-1-i}\right]\right)
	&\leq \sum_{w=2^{m-2-i}}^{2^{m-1-i}} \text{exp}_2\left(2^{m-1}\cdot {R_{m,+}}\cdot h_b\left(\frac{w}{2^{m-1}}\right)\right)\notag\\
	&\stackrel{(e)}{\leq} \sum_{w=2^{m-2-i}}^{2^{m-1-i}} \text{exp}_2\left(2^{m-1}\cdot {R_{m,+}}\cdot h_b(2^{-i})\right) \notag\\
	&= \text{exp}_2\left(2^{m-1}\cdot {R_{m,+}}\cdot h_b(2^{-i}) + o(2^m)\right):=B_i(m), \label{eq:sub2}
\end{align}
where inequality (e) above follows from the fact that the binary entropy function, $h_b(p)$, is increasing for arguments $p\in \left[0,\frac{1}{2}\right)$.
Therefore, putting \eqref{eq:sub2} back in \eqref{eq:sub1}, we get that
\begin{align}
	2^m{\mathsf{R}^{(1)}_{m}(\mathcal{C})}
	& = \log_2\left\lvert{\mathcal{H}^{(1)}_m} \right\rvert \notag\\
	&\leq {m-1 \choose \leq r_m}+
	 \log_2\Bigg\{\frac{1}{2}\cdot \text{exp}_2{m-2 \choose \leq r_m}+ \sum_{i=1}^{r_m-1}B_i(m)\cdot \text{exp}_2\left( -{m-2-i \choose \leq r_m}\right)\Bigg\} \notag\\
	  &= {m-1 \choose \leq r_m}+\log_2\left(\alpha(m)+\beta(m)\right),\label{eq:subs3}
\end{align}
where we define
\begin{align}
	\alpha(m)&:= \frac12 \cdot \text{exp}_2{m-2 \choose \leq r_m}, \ \text{and} \notag\\
	\beta(m)&:= \sum_{i=1}^{r_m-1}B_i(m)\cdot \text{exp}_2\left( -{m-2-i \choose \leq r_m}\right). \label{eq:beta}
\end{align}
In Appendix A, we show that for all $\delta\in (0,1)$ sufficiently small and for $m$ sufficiently large, we have
\begin{align}
\beta(m) \ &\le \ \text{exp}_2\left(2^{m-1}R\cdot\left(\frac34 +\delta\right)+o(2^m)\right) \notag\\
& =: \ \theta(m).
\label{eq:inter}
\end{align}
Now, using Lemma~\ref{lem:rate}, we have
\[
\lim_{m\to \infty} \frac{1}{2^m} {m-2 \choose \leq r_{m}} = \frac{R}{4}.
\]
Hence, for small $\delta\in (0,1)$, and for $m$ large enough,
\[
{m-2 \choose \leq r_{m}}\leq (1+\delta)\cdot 2^{{m}-2}\cdot R.
\]
Therefore, we get that
\begin{align}
\alpha(m) \ \leq \ \text{exp}_2\left((1+\delta)\cdot2^{{m}-2}\cdot R\right) \ =: \ \eta(m). \label{eq:inter2}
\end{align}
Now, substituting \eqref{eq:inter} and \eqref{eq:inter2} in \eqref{eq:subs3}, we get that 
\begin{align}
2^m{\mathsf{R}^{(1)}_{m}(\mathcal{C})}\leq {m-1 \choose \leq r_m}+\log_2\left(\eta(m)+\theta(m)\right).
\label{eq:inter3}
\end{align}
Putting everything together, we see that
\begin{align}
	{\mathsf{R}^{(1)}(\mathcal{C})} &= \limsup_{m\to \infty}{\mathsf{R}^{(1)}_{m}(\mathcal{C})}\notag\\
	&\leq \lim_{m\to \infty} \frac{1}{2^m} \left[{m-1 \choose \leq r_m}+\log_2\left(\eta(m)+\theta(m)\right)\right] \notag\\
	&\stackrel{(p)}{\leq} \lim_{m\to \infty} \frac{1}{2^m} \left[{m-1 \choose \leq r_m}+\log_2\left(2\cdot \theta(m)\right)\right] \notag\\
	&\stackrel{(q)} =\frac{R}{2} + \lim_{m\to \infty} \frac{1}{2^m} \cdot \log_2 \theta(m) \notag \\
	&= \frac{R}{2}+\frac{3R}{8}+\frac{\delta R}{2}\notag\\
	&= \frac{7R}{8}+\frac{\delta R}{2} \label{eq:final}.
\end{align}
Note that inequality (p) follows from the fact for any $R \in (0,1)$, $\eta(m)\leq \theta(m)$ holds for all sufficiently small $\delta > 0$. Further, equation (q) is valid because $\lim_{m\to \infty} \frac{1}{2^m} {m-1 \choose \leq r_m} = \frac{R}{2}$, by Lemma~\ref{lem:rate}. Since \eqref{eq:final} holds for all $\delta>0$ sufficiently small, we can let $\delta\to 0$, thereby obtaining that
\begin{align*}
	{\mathsf{R}^{(1)}(\mathcal{C})} \le \frac{7R}{8}.
\end{align*}
Moreover, since for any $m\geq 1$, we have that ${\mathcal{H}^{(1)}_m} \subseteq S_{(1,\infty)}^{(2^m)}$, with $\lim_{m\to \infty} \frac{\log_2 \left\lvert S_{(1,\infty)}^{(2^m)} \right \rvert}{2^m} = {\kappa_1}$, the inequality ${\mathsf{R}^{(1)}(\mathcal{C})} \leq {\kappa_1}$ holds.
\end{proof}

{\subsection{Linear $(d,\infty)$-RLL Constrained Subcodes Under Alternative Coordinate Orderings}}
\label{sec:perm}

Throughout the previous subsections, we have assumed that the coordinates of the RM codes are ordered according to the standard lexicographic ordering. In this subsection, we shall address the question of whether alternative coordinate orderings allow for larger rates of \emph{linear} $(d,\infty)$-RLL constrained subcodes of RM codes of rate $R$, as compared to the upper bound of $\frac{R}{d+1}$ derived for RM codes under the lexicographic coordinate ordering, in Theorem \ref{thm:rmlinub}.

First, we consider a Gray ordering of coordinates of the code RM$(m,r)$. In such an ordering, consecutive coordinates $\mathbf{b} = (b_1,\ldots,b_m)$ and $\mathbf{b}^\prime = (b_1^\prime,\ldots,b_m^\prime)$ are such that for some bit index $i\in [m]$, $b_i\neq b_i^\prime$, but $b_j = b_j^\prime$, for all $j\neq i$. In words, consecutive coordinates in a Gray ordering, when represented as $m$-tuples, differ in exactly one bit index. Note that there are multiple orderings that satisfy this property. Indeed, Gray orderings correspond to Hamiltonian paths (see, for example, \cite{diestel}, Chap. 10) on the $m$-dimensional unit hypercube. 

In what follows, we work with a fixed sequence of Gray orderings defined as follows: let $(\pi_{\text{G},m})_{m\geq 1}$ be a sequence of permutations, with $\pi_{\text{G},m}: [0:2^m-1]\rightarrow [0:2^m-1]$, for any $m\geq 1$, having the property that \textbf{B}$(\pi_{\text{G},m}(j))$ differs from \textbf{B}$(\pi_{\text{G},m}(j-1))$ in exactly one bit index, for any $j\in [1:2^m-1]$.

Now, as before, fix a sequence of codes $\left\{{\hat{\mathcal{C}}_m} = \text{RM}(m,r_m)\right\}_{m\geq 1}$ of rate $R\in (0,1)$. Note that for large enough $m$, we have that $r_m\leq m-1$. We again use the notation $K_m:={m\choose \leq r_m}$ to denote the dimension of ${{\hat{\mathcal{C}}_m}}$. We then define the sequence of Gray-ordered RM codes $\left\{{\mathcal{C}_m^{\text{G}}}\right\}_{m\geq 1}$, via
\begin{align*}{\mathcal{C}_{m}^\text{G}}:= \big\{(c_{\pi_{\text{G},m}(0)},c_{\pi_{\text{G},m}(1)},\ldots,c_{\pi_{\text{G},m}(2^m-1)}):(c_0,c_1,\ldots,c_{2^m-1})\in {\hat{\mathcal{C}}_m}\big\}.\end{align*}
Clearly, the sequence of codes $\left\{{\mathcal{C}_m^{\text{G}}}\right\}_{m\geq 1}$ is also of rate $R\in (0,1)$. In order to obtain an upper bound on the rate of the largest linear $(d,\infty)$-RLL subcode of the code ${\mathcal{C}_m^{\text{G}}}$, as in Theorem \ref{thm:rmlinub}, we shall again work with the information set $\mathcal{I}_{m,r_m}$ (see \eqref{eq:Imr} for the definition of the set $\mathcal{I}_{m,r}$ of coordinates of RM$(m,r)$). 

Analogous to \eqref{eq:gammam}, we define the set
\begin{align}
	\label{eq:gammamgray}
	\Gamma_{m}^\text{G}:=\{s\in [0:2^m-1]:\ \textbf{B}(\pi_{\text{G},m}(s))\in \mathcal{I}_{m,r_m},\ \text{but}\  \textbf{B}(\pi_{\text{G},m}(s+1))\notin \mathcal{I}_{m,r_m}\}.
\end{align}
{to be the collection of right end-points of runs of coordinates in the information set $\mathcal{I}_{m,r_m}$, where the coordinates are now ordered according to the fixed Gray ordering.} We use the convention that when $s=2^m-1$,  $\textbf{B}(\pi_{\text{G},m}(s+1))$ is defined to be a dummy symbol `$\times$' that does not belong to $\mathcal{I}_{m,r_m}$. The number of such runs is $\left \lvert \Gamma_{m}^\text{G}\right \rvert$. Note now that unlike in \eqref{eq:gammam}, it is possible that $2^{m}-1\in \Gamma_{m}^\text{G}$, since, depending on the specific Gray coordinate ordering chosen, it is possible that $\textbf{B}(\pi_{\text{G},m}(2^m-1)) \in \mathcal{I}_{m,r_m}$.

We now state and prove a lemma analogous to Lemma \ref{lem:runs}:
\begin{lemma}
	\label{lem:runsgray}
	Under a fixed Gray ordering defined by $\pi_m^\text{G}$, the inequality $\left \lvert\Gamma_{m,r_m}^\text{G} \right \rvert \leq {m\choose r_m+1}+1$ holds, for $0\leq r_m\leq m-1$.
\end{lemma}
\begin{proof}
	For any $s\in [0:2^m-2]$ that belongs to $\Gamma_{m}^\text{G}$, we have wt$(\mathbf{B}(\pi_{\text{G},m}(s)))\leq r_m$, but wt$(\mathbf{B}(\pi_{\text{G},m}(s+1)))\geq r_m+1$. In fact, since consecutive coordinates differ in exactly one bit index in the Gray ordering, it must be the case that wt$(\mathbf{B}(\pi_{\text{G},m}(s+1)))= r_m+1$. Thus, the number of integers $s\in [0:2^m-2]$ belonging to $\Gamma_{m}^\text{G}$ is bounded above by ${m\choose r_m+1}$, which is the number of coordinates whose binary representation has weight exactly $r_m+1$. In order to account for the possibility that $2^m-1 \in \Gamma_{m}^\text{G}$, we state the overall upper bound on the number of runs as ${m\choose r_m+1}+1$.
\end{proof}

With Lemma \ref{lem:runsgray} established, we now {prove} Theorem \ref{thm:grayinf}. {Our proof strategy is similar to that for Theorem \ref{thm:rmlinub}: for a given code ${\mathcal{C}_m^{\text{G}}}$, we first calculate a lower bound on the number of $(d+1)$-tuples of consecutive coordinates, ordered according to the fixed Gray ordering, in the information set $\mathcal{I}_{m,r_m}$. Via Proposition \ref{prop:linub}, this gives us an upper bound on the rate of any linear $(d,\infty)$-RLL subcode of ${\mathcal{C}_m^{\text{G}}}$. We then let the blocklength go to infinity, to obtain our desired result.}
\begin{proof}[Proof of Theorem \ref{thm:grayinf}]
	Similar to the proof of Theorem \ref{thm:rmlinub}, the calculation of the overall number, $t_m^\text{G}$, of disjoint $(d+1)$-tuples of consecutive coordinates in $\mathcal{I}_{m,r_m}$, for large enough $m$, results in
	\begin{align*}
		t_m^\text{G} \geq \frac{K_m}{d+1}-{m\choose r_m+1}-1.
	\end{align*}
	{Note that here too, as in the proof of Theorem \ref{thm:rmlinub}, we use the notation $K_m$ to denote the dimension of ${{\mathcal{C}_m^\text{G}}}$.} Again, using Proposition \ref{prop:linub}, it follows that the dimension of any linear $(d,\infty)$-RLL subcode ${\overline{\mathcal{H}}_{\text{G}}^{(d)}} \subseteq {{\mathcal{C}_m^\text{G}}}$ is at most $K_m-dt_m^\text{G}$. Now, recalling the definition of ${\overline{\mathsf{R}}^{(d)}(\mathcal{C}^\text{G})}$, from equation \eqref{eq:Rubgray}, we see that
	\begin{align*}
		{\overline{\mathsf{R}}^{(d)}(\mathcal{C}^\text{G})}&\leq \limsup_{m\to \infty}\frac{K_m-dt_m^\text{G}}{2^m}\\
		&\leq \limsup_{m\to \infty}\frac{K_m-\frac{dK_m}{d+1}+d\cdot {m\choose r_m+1}+d}{2^m}\\
		&\leq \lim_{m\to \infty}\frac{\frac{K_m}{d+1}+d\cdot {m\choose \left \lfloor \frac{m}{2}\right \rfloor}+d}{2^m}\\
		&= \frac{R}{d+1},
	\end{align*}
	where the last equality holds for reasons similar to those in the proof of Theorem \ref{thm:rmlinub}. 
\end{proof}

Next, we shift our attention to $\pi$-ordered RM codes $\left\{{\mathcal{C}_m^\pi}\right\}_{m\geq 1}$, defined by the sequence of arbitrary permutations $(\pi_m)_{m\geq 1}$, with $\pi_m: [0:2^m-1]\to [0:2^m-1]$ (see the discussion preceding Theorem \ref{thm:genperminf} in Section \ref{sec:main}). {Note that the code ${\mathcal{C}_m^\pi}$ has the same dimension $K_m$. as the code ${\hat{\mathcal{C}}_m = \text{RM}(m,r_m)}$.} Also recall the definition of ${\overline{\mathcal{H}}_{\pi}^{(d)}}$ being the largest \emph{linear} $(d,\infty)$-RLL subcode of ${\mathcal{C}_m^\pi}$.

We shall now prove Theorem \ref{thm:genperminf}. {The proof goes as follows: we first show that for large $m$, for almost all coordinate permutations $\pi_m$, the first block of $K_m(1+\alpha_m)$ coordinates in the $\pi$-ordered code ${\mathcal{C}_m^\pi}$ contains an information set $\mathcal{J}_{m,r_m}$, where $\alpha_m \to 0$ as $m \to \infty$. This then allows us to arrive at a lower bound on the number of disjoint $(d+1)$-tuples of consecutive coordinates in this information set $\mathcal{J}_{m,r_m}$. Again, using Proposition \ref{prop:linub}, we arrive at the desired upper bound on the rate of any $(d,\infty)$-RLL constrained linear subcode of ${\mathcal{C}_m^\pi}$, for almost all permutations $\pi_m$.}
\begin{proof}[Proof of Theorem \ref{thm:genperminf}]
	We wish to prove that for ``most'' orderings, and for large $m$, the rate of ${\overline{\mathcal{H}}_{\pi}^{(d)}}$ is bounded above by $\frac{R}{d+1}+{\epsilon_m}$, where ${\epsilon_m} \xrightarrow{m\to \infty}0$.
	
	To this end, we first make the observation that the sequence of RM codes $\left\{{\hat{\mathcal{C}}_m} = \text{RM}(m,r_m)\right\}_{m\geq 1}$ achieves a rate $R$ over the BEC, under block-MAP decoding too (see \cite{kud1} and \cite{kud3}). Hence, for large enough $m$, the (linear) RM code ${\hat{\mathcal{C}}_m}$ can correct erasures that are caused by a BEC$(1-R-\gamma_m)$, with $\gamma_m>0$, and $\gamma_m\xrightarrow{m\to \infty}0$. This then means that for large $m$, ${\hat{\mathcal{C}}_m}$ can correct $2^m(1-R-\gamma_m)-c\cdot \sqrt{2^m (1-R-\gamma_m)}$ erasures, with high probability (see Lemma 15 of \cite{abbe1}), for $c>0$ suitably small. Finally, from Corollary 18 of \cite{abbe1}, it then holds that for large enough $m$, any collection of $2^m R(1+\beta_m)$ columns of $G_\text{Lex}(m,r_m)$, chosen uniformly at random, must have full row rank, $K_m{:= {m\choose \leq r_m}}$, with probabilty $1-\delta_m$, with $\beta_m, \delta_m>0$ and $\beta_m,\delta_m\xrightarrow{m\to \infty}0$.

	In other words, the discussion above implies that for large enough $m$, a collection of $K_m(1+\alpha_m)$ coordinates, chosen uniformly at random, contains an information set, with probabilty $1-\delta_m$, where, again, $\alpha_m>0$, with $\alpha_m\xrightarrow{m\to \infty}0$. An equivalent view of the above statement is that for large enough $m$, for a $1-\delta_m$ fraction of the possible permutations $\pi_m: [0:2^m-1]\to [0:2^m-1]$, the first block of $K_m(1+\alpha_m)$ coordinates of the code ${\mathcal{C}_{m}^\pi}$, contains an information set, ${\mathcal{J}}_{m,r_m}${, with dim$({\mathcal{C}_{m}^\pi}) = K_m$}. Now, within these ``good'' permutations, since $|{\mathcal{J}}_{m,r_m}| = K_m$, it follows that the number of runs, $\left\lvert\Gamma_{m}^\pi\right \rvert$, of consecutive coordinates that belong to ${\mathcal{J}}_{m,r_m}$, obeys $\left\lvert\Gamma_{m}^\pi\right \rvert \leq K_m\alpha_m+1$, with $\Gamma_{m}^\pi$ defined similar to equation \eqref{eq:gammamgray}. This is because the number of runs, $\left\lvert\Gamma_{m}^\pi\right \rvert$, equals the number of coordinates $s$,  such that $s\in {\mathcal{J}}_{m,r_m}$, but $s+1\notin {\mathcal{J}}_{m,r_m}$. Hence, accounting for the possibility that the last coordinate in the $K_m(1+\alpha_m)$-block belongs to ${\mathcal{J}}_{m,r_m}$, we get that the number of such coordinates $s$ is at most $K_m(1+\alpha_m)+1-K_m$, which equals $K_m\alpha_m+1$.
	
	Hence, the overall number, $t_m^\pi$, of disjoint $(d+1)$-tuples of consecutive coordinates in $\mathcal{J}_{m,r}$, satisfies (see the proof of Theorem \ref{thm:rmlinub})
	\[
	t_m^\pi\geq \frac{K_m}{d+1}-K_m\alpha_m-1,
	\]
	for a $1-\delta_m$ fraction of permutations $\pi_m$. Again, applying Proposition \ref{prop:linub}, we have that for a $1-\delta_m$ fraction of permutations, with $\delta_m\xrightarrow{m\to \infty}0$, the rate of the largest $(d,\infty)$-RLL subcode obeys, for $m$ large,
	\begin{align*}
		\frac{\log_2\left \lvert {\overline{\mathcal{H}}_{\pi}^{(d)}}\right\rvert}{2^m}&\leq \frac{K_m-dt_m^\pi}{2^m}\\
		&\leq \frac{K_m-\frac{dK_m}{d+1}+dK_m\alpha_m+d}{2^m}\\
		&= \frac{R}{d+1}+\epsilon_m,
	\end{align*}
	where $\epsilon_m\xrightarrow{m\to \infty}0$, {with the last equality following from the fact that $\lim_{m\to \infty} \frac{K_m}{2^m} = R$. The theorem thus follows.}
\end{proof}
\section{A Two-Stage Constrained Coding Scheme}
\label{sec:cosets}
The results summarized in the previous sections provide lower and upper bounds on rates of $(d,\infty)$-RLL constrained {subcodes} of RM codes of rate $R$. In particular, Theorem \ref{thm:rm} identifies linear subcodes of RM codes of rate $2^{-\left \lceil \log_2(d+1)\right \rceil} \cdot R$, and for the special case when $d=1$, Theorem \ref{thm:nonlin} improves upon this lower bound, for a range of $R$ values, by proving the existence of (non-linear) subcodes of rate at least $\max\left(0,R-\frac38\right)$. Furthermore, using the bit-MAP or block-MAP decoders corresponding to the parent RM codes, we observe that these rates are achievable over $(d,\infty)$-RLL input-constrained BMS channels, so long as $R<C$, where $C$ is the capacity of the unconstrained BMS channel. In this section, we provide another explicit construction of $(d,\infty)$-RLL constrained codes via a concatenated (or two-stage) coding scheme, the outer code of which is a systematic RM code of rate $R$, and the inner code of which employs the $(d,\infty)$-RLL constrained subcodes identified in Theorem \ref{thm:rm}. The strategy behind our coding scheme is very similar to the ``reversed concatenation'' scheme that is used to limit error propagation while decoding constrained codes over noisy channels (see \cite{Bliss, Mansuripur, Imm97, FC98} and Section 8.5 in \cite{Roth}). However, to keep the exposition self-contained, we describe the scheme from first principles here. We then derive a rate lower bound for this scheme (see Theorem \ref{thm:rmcosets}), and prove that this lower bound is achievable, under \emph{block-MAP} decoding over $(d,\infty)$-RLL input-constrained BMS channels, if $R<C$. Hence, the lower bound given in Theorem \ref{thm:rmcosets} is achievable, when $R\in (0,1)$ is replaced by $C$. {Note that here we use the fact that RM codes achieve any rate $R<C$, under block-MAP decoding (see \cite{abbesandon}).}

We now describe our two-stage coding scheme. Fix a rate $R\in (0,1)$ and any sequence $\left\{{\hat{\mathcal{C}}_m} = \text{RM}(m,r_m)\right\}_{m\geq 1}$ of RM codes of  rate $R$. We interchangeably index the coordinates of any codeword in ${\hat{\mathcal{C}}_m}$ by $m$-tuples in the lexicographic order, and by integers in $[0:2^m-1]$. Recall, from Lemma \ref{lem:infoset} and \eqref{eq:Imr}, that the set $\mathcal{I}_{m,r_m}:= \{\mathbf{b} = (b_1,\ldots,b_m)\in \mathbb{F}_2^m: \text{wt}(\mathbf{b})\leq r_m\}$ is an information set of ${\hat{\mathcal{C}}_m}$. For the remainder of this section, we let $m$ be a large positive integer.

{We first set up some notation.} Let $K_m=$ dim$\left({\hat{\mathcal{C}}_m}\right) = {m\choose \leq r_m}$ and $N_m:=2^m$ and note that $N_m-K_m = (1+\beta_m)N_m(1-R)$, where $\beta_m$ is a correction term that vanishes as $m\to \infty$.
For the outer code, or the first stage, of our coding scheme, we shall work with an RM code in systematic form, in which the first $K_m$ positions form the information set $\mathcal{I}_{m,r_m}$. Consider any permutation $\pi_m: [0:N_m-1]\to [0:N_m-1]$ with the property that $\pi_m([0:K_m-1]) = \mathcal{I}_{m,r_m}$, where, for a permutation $\sigma$, and a set $\mathcal{A}\subseteq [0:N_m-1]$, we define the notation $\sigma(\mathcal{A}):=\{\sigma(i):i\in \mathcal{A}\}$. We then define the equivalent systematic RM code ${\mathcal{C}^{\pi}_m}$ as
\begin{align*}
{{\mathcal{C}}^{\pi}_m} = \big\{(c_{\pi_m(0)},c_{\pi_m(1)},\ldots,c_{\pi_m(N_m-1)}):(c_0,c_1,\ldots,c_{N_m-1})\in {\hat{\mathcal{C}}_m}\big\}.
\end{align*}
We let $G^\pi_m$ be a \emph{systematic} generator matrix for ${\mathcal{C}^{\pi}_m}$. We then recall the definition of the subcode ${\mathcal{C}_n^{(d)}}$ of the code ${\mathcal{C}_n}$ (see \eqref{eq:rmval} and the proof of Theorem \ref{thm:rm}). We let a generator matrix of the linear code ${\mathcal{C}_n^{(d)}}$ be denoted by $G_{n}^{(d)}$. 

\subsection*{{Encoding}}

{Our encoding algorithm is shown as Algorithm \ref{alg:coset}, where the values of the parameters $L$ and $n^\star$ will be specified later.}
As mentioned earlier, our $(d,\infty)$-RLL constrained coding scheme comprises two stages: an outer encoding stage (\textbf{E1}) and an inner encoding stage (\textbf{E2}) as given below.
\begin{enumerate}
	\item[(\textbf{E1})] Pick a $(d,\infty)$-RLL constrained $K_m$-tuple,  {$\mathbf{w}$}. Encode {$\mathbf{w}$} into a codeword $\mathbf{c} \in {\mathcal{C}_m^\pi}$, using the systematic generator matrix $G_m^\pi$: $ \mathbf{c} = {\mathbf{w}} \cdot G_m^\pi$. Note that $c_0^{K_m-1} = {\mathbf{w}}$ is $(d,\infty)$-RLL constrained. {This is shown in Steps 2 and 3 in Algorithm \ref{alg:coset}. Note that choosing an RLL constrained word in Step (\textbf{E1}) above can be accomplished using well-known constrained encoders (see, for example, \cite{adler} and Chapters 4 and 5 of \cite{Roth}), of rates arbitrarily close to the noiseless capacity ${\kappa_d}$ of the $(d,\infty)$-RLL constraint.}
	\item[(\textbf{E2})] Encode the last $N_m-K_m$ bits, $c_{K_m}^{N_m-1}$, of $\mathbf{c}$, using $(d,\infty)$-RLL constrained codewords of RM codes of rate $R$, {as shown in Steps 5--7 in Algorithm~\ref{alg:coset}. In what follows, we elaborate on the choices of $L$ and $n^\star$ in this part of the algorithm.}
\end{enumerate}
\begin{algorithm}[t]
	\caption{Construction of $(d,\infty)$-RLL constrained code $\mathcal{C}^{\text{conc}}_m$}
	\label{alg:coset}
	\begin{algorithmic}[1]	
		\Procedure{Coding-Scheme}{$G^{\pi}_m$, $G_{n^\star}^{(d)}$}       
		\State Pick a $(d,\infty)$-RLL constrained $K_m$-tuple {$\mathbf{w}$}
		\State Obtain $\mathbf{c} \in {\mathcal{C}_m^\pi}$ as $\mathbf{c} = {\mathbf{w}} \cdot G_m^\pi$.
		\State Set $\mathbf{x}_1:= {\mathbf{w}}$.
		\State Divide $c_{K_m}^{N_m-1}$ into $L$ equal-length blocks, ${\mathbf{c}}_1,\ldots, {\mathbf{c}}_L$.
		\For{$i=1:L$}
		\State Set $\mathbf{x}_{2,i} = {\mathbf{c}}_i \cdot G_{n^\star}^{(d)}$.
		\EndFor
		\State Set $\mathbf{x}_2=\mathbf{x}_{2,1}\ldots \mathbf{x}_{2,L}$.
		\State Transmit $\mathbf{x} = \mathbf{x}_1\mathbf{x}_2$.
		\EndProcedure
		
	\end{algorithmic}
\end{algorithm} 

{The main idea is to encode $c_{K_m}^{N_m-1}$ using the family of linear $(d,\infty)$-RLL subcodes $\bigl\{\mathcal{C}_n^{(d)}\bigr\}_{n\geq 1}$ of rate-$R$ RM codes, given by Theorem \ref{thm:rm}. Recall that these subcodes achieve rate $2^{-\left \lceil \log_2(d+1)\right \rceil}\cdot R$ as $n \to \infty$. So, encoding the $N_m-K_m$ bits in $c_{K_m}^{N_m-1}$ using a code $\mathcal{C}_n^{(d)}$ will result in an encoded blocklength of roughly $\bigl(\frac{N_m-K_m}{R}\bigr) \, 2^{\left \lceil \log_2(d+1)\right \rceil}$ bits. However, the codes $\mathcal{C}_n^{(d)}$ have blocklength equal to $2^n$, so the $\bigl(\frac{N_m-K_m}{R}\bigr) \, 2^{\left \lceil \log_2(d+1)\right \rceil}$ encoded bits must be formed by concatenating an integer number of codewords of blocklength $2^n$. In other words, $\bigl(\frac{N_m-K_m}{R}\bigr) \, 2^{\left \lceil \log_2(d+1)\right \rceil}$ must be an integer multiple of a power of $2$. Writing $N_m - K_m \approx N_m(1-R)$ and recalling that $N_m = 2^m$, we observe that $\bigl(\frac{N_m-K_m}{R}\bigr) \, 2^{\left \lceil \log_2(d+1)\right \rceil}$ is (approximately) expressible as $\bigl(\frac{1-R}{R}\bigr) 2^{m + \left \lceil \log_2(d+1)\right \rceil}$. For this to be an integer multiple of a power of $2$, $\frac{1-R}{R}$ should be well-approximated by a dyadic rational of the form $\frac{L}{2^\tau}$. Then, choosing $n^{\star} = m-\tau + \left \lceil \log_2(d+1)\right \rceil$, we obtain that $\bigl(\frac{1-R}{R}\bigr) 2^{m + \left \lceil \log_2(d+1)\right \rceil}$ is (approximately) equal to $L \cdot 2^{n^\star}$. Thus, it should be possible to encode the $N_m-K_m$ bits $c_{K_m}^{N_m-1}$ by first chopping it up into $L$ equal-length blocks, and encoding each block using the code $\mathcal{C}_{n^\star}^{(d)}$. We formalize this argument below.}

{Pick an arbitrarily small $\epsilon > 0$, and choose large positive integers $m_0$ and $\tau$, and a positive integer $L$, such that
\begin{align}
	\label{eq:L}
	\frac{(1-R)(1+\beta_m)}{R(1-\epsilon)} \subseteq \left[\frac{L-1}{2^\tau}, \frac{L}{2^\tau}\right],\ \text{for all $m\geq m_0$}.
\end{align}
Note that $L$ and $\tau$, though large, are constants.
We then set $n^\star:= m-\tau+\left \lceil \log_2(d+1)\right \rceil$. Now, partition the $N_m-K_m$ bits, $c_{K_m}^{N_m-1}$, into $L$ blocks $\mathbf{c}_1,\ldots,\mathbf{c}_L$, each $\mathbf{c}_i$ having $\frac{N_m-K_m}{L}$ bits\footnote{For ease of description, we assume that $m$ is such that $L$ divides $N_m-K_m$. The general case can be handled by appending at most $L-1$ $0$s at the end of the $N_m-K_m$ bits, so that the overall length is divisible by $L$, thereby giving rise to the same lower bound in {Theorem \ref{thm:rmcosets}}.}. As indicated by Step~7 of Algorithm~\ref{alg:coset}, to encode each $\mathbf{c}_i$, $i = 1,\ldots,L$, we use a code $\mathcal{C}_{n^\star}^{(d)}$ from the family of linear $(d,\infty)$-RLL RM subcodes $\bigl\{{\mathcal{C}_n^{(d)}}\bigr\}_{n\geq 1}$ of rate $2^{-\left \lceil \log_2(d+1)\right \rceil}\cdot R$ given by Theorem \ref{thm:rm}. We choose $n^\star$ large enough (by taking $m$ large enough) that the rate of the subcode ${\mathcal{C}_{n^\star}^{(d)}}$ is at least $2^{-\left \lceil \log_2(d+1)\right \rceil}\cdot R (1-\epsilon)$.
With this, the dimension of the code $\mathcal{C}_{n^\star}^{(d)}$ is 
\begin{equation}
\dim\left(\mathcal{C}_{n^\star}^{(d)}\right) \ge 
2^{n^\star-\left \lceil \log_2(d+1)\right \rceil}\cdot R (1-\epsilon)
= 2^{m-\tau} \cdot R(1-\epsilon).
\label{ineq:dimCnd}
\end{equation}
From \eqref{eq:L}, we find that $2^{-\tau}\cdot R(1-\epsilon) \ge \frac{1}{L} (1-R)(1+\beta_m)$, so that carrying on from \eqref{ineq:dimCnd}, we have
$$
\dim\left(\mathcal{C}_{n^\star}^{(d)}\right)\ge \frac{1}{L} N_m (1-R)(1+\beta_m) = \frac{N_m-K_m}{L}.
$$
This means that each block $\mathbf{c}_i$ can indeed be encoded into a unique codeword of $\mathcal{C}_{n^\star}^{(d)}$, as encapsulated in Step~7 of Algorithm~\ref{alg:coset}.\footnote{{It may be necessary to pad $\mathbf{c}_i$ with some extra $0$s to make its blocklength match the dimension of $\mathcal{C}_{n^\star}^{(d)}$.}}
Thus, each $\mathbf{c}_i$ is encoded into a codeword of $\mathcal{C}_{n^\star}^{(d)}$, having blocklength $N_{\text{part}}:=2^{n^\star}$. Hence, the total encoded blocklength for all the blocks $\mathbf{c}_1,\ldots,\mathbf{c}_L$ is $N_{\text{part}}\cdot L$, and the total number of channel uses for transmission (see Step 9 of Algorithm \ref{alg:coset}) is $N_{\text{tot}}:=K_m+N_{\text{part}}\cdot L$. 
}

{Moreover, we note from the construction of ${\mathcal{C}_n^{(d)}}$ in the proof of Theorem \ref{thm:rm} in Section \ref{sec:rm} that the first $d$ symbols in $\mathbf{x}_{2,i}$ are $0$s, for all $i\in [L]$. Hence, the $(d,\infty)$-RLL input constraint is also satisfied at the boundaries of the concatenations in Steps 8 and 9 of Algorithm 1.}

\subsection*{{Decoding}}

Since we intend transmitting the $(d,\infty)$-RLL constrained code $\mathcal{C}^{\text{conc}}_m$ over a noisy BMS channel, we now specify the decoding strategy. Let $y_0^{N_\text{tot}-1}$ be the vector of symbols received by the decoder. Decoding is, as encoding was, a two-stage procedure. In the first stage, the {block}-MAP decoder of the code ${\mathcal{C}_{n^\star}} \supseteq {\mathcal{C}_{n^\star}^{(d)}}$, is used for each of the $L$ parts $\mathbf{c}_1,\ldots,\mathbf{c}_L$, to obtain the estimate $\hat{c}_{K_m}^{N_m-1}:=(\hat{c}_{K_m},\ldots,\hat{c}_{N_m-1})$ of the last $N_m-K_m$ bits $c_{K_m}^{N_m-1}$. In the second stage, the {block}-MAP decoder of the systematic RM code ${\mathcal{C}}^{\pi}_m(R)$ takes as input the vector $y_0^{K_m-1}\hat{c}_{K_m}^{N_m-1}$, and produces as (the final) estimate, $\hat{c}_0^{K_m-1} := (\hat{c}_0\ldots,\hat{c}_{K_m-1})$, of the information bits $c_0^{K_m-1} = \mathbf{w}$. The decoding strategy is summarized below.

\begin{itemize}
	\item [(\textbf{D1})] Decode each of the $L$ parts $\mathbf{c}_1,\ldots,\mathbf{c}_L$, using the {block}-MAP decoder of ${\mathcal{C}_{n^\star}}$, to obtain the estimate $\hat{c}_{K_m}^{N_m-1}$.
	\item [(\textbf{D2})] Using the vector $y_0^{K_m-1}\hat{c}_{K_m}^{N_m-1}$ as input to the {block}-MAP decoder of ${\mathcal{C}_m^\pi}$, obtain the estimate $\hat{c}_0^{K_m-1}$ of the information bits ${\mathbf{w}}$.
\end{itemize} 
%
%

A rate lower bound of the coding scheme in Algorithm \ref{alg:coset} and a lower bound on the probability of correct decoding is provided in the proof of Theorem \ref{thm:rmcosets} below.
\begin{proof}[Proof of Theorem \ref{thm:rmcosets}]
	Consider the code $\mathcal{C}_m^\text{conc}$, given in Algorithm 1, for large values of $m$, and the decoding procedure given in (\textbf{D1})--(\textbf{D2}).
	By picking $m$ large enough (and hence $K_m$ large enough), we note that for Step 2 of Algorithm 1, there exist constrained coding schemes (see \cite{adler} and Chapters 4 and 5 of \cite{Roth}) of rate ${\kappa_d}- \alpha_m$, for $\alpha_m>0$, with $\alpha_m\xrightarrow{m\to \infty} 0$. Hence, we see that for large $m$, the number of possible $K_m$-tuples ${\mathbf{w}}$ that can be picked, equals $2^{K_m({\kappa_d}-\alpha_m)}$. Since the codeword $\mathbf{c}$ and the words $\mathbf{x}_1$ and $\mathbf{x}_2$ are determined by $\mathbf{w}$, we have that for large $m$, the rate of the code $\mathcal{C}_m^{\text{cos}}$ obeys
	\begin{align*}
	\label{eq:ratecalc}
	\text{rate}(\mathcal{C}_m^{\text{conc}}) \geq \frac{\log_2\left(2^{K_m({\kappa_d}-\alpha_m)}\right) }{K_m+N_{\text{part}}\cdot L},
\end{align*}
where the denominator, $K_m+N_{\text{part}}\cdot L$, is the total number of channel uses. The following statements then hold true:
\begin{align*}
	\text{rate}(\mathcal{C}_m^{\text{conc}}) &\geq \frac{\log_2\left(2^{K_m({\kappa_d}-\alpha_m)}\right)}{K_m+N_{\text{part}}\cdot L}\\
	&\stackrel{(a)}{=} \frac{\frac{\left({\kappa_d}-\alpha_m\right)\cdot{K_m}}{N_m}}{\frac{K_m}{N_m}+L\cdot 2^{-\tau+\left \lceil \log_2(d+1)\right \rceil}}\\
	&\stackrel{(b)}{\geq} \frac{\frac{\left({\kappa_d}-\alpha_m\right)\cdot{K_m}}{N_m}}{\frac{K_m}{N_m}+2^{\left \lceil \log_2(d+1)\right \rceil}\cdot \left(\frac{(1-R)(1+\beta_m)}{R(1-\epsilon)}+2^{-\tau}\right)},
\end{align*}
where (a) follows from the definition of {$N_{\text{part}}$} and (b) holds due to equation \eqref{eq:L}, with $L\cdot 2^{-\tau}\leq \left(\frac{(1-R)(1+\beta_m)}{R(1-\epsilon)}+2^{-\tau}\right)$. Hence, by taking $\liminf_{m\to \infty}$ on both sides of the inequality (b) above, we get
\begin{align*}
	\liminf_{m\to \infty} \text{rate}(\mathcal{C}_m^{\text{cos}}) &\geq \frac{{\kappa_d}\cdot R}{R+2^{\left \lceil \log_2(d+1)\right \rceil}\cdot \left(\frac{1-R}{R(1-\epsilon)}\right)+2^{\left \lceil \log_2(d+1)\right \rceil-\tau}}\\
	&= \frac{(1-\epsilon)\cdot {\kappa_d}\cdot R^2\cdot 2^{-\left \lceil \log_2(d+1)\right \rceil}}{(1-\epsilon)R^2\cdot 2^{-\left \lceil \log_2(d+1)\right \rceil} + 1-R+R\cdot 2^{-\tau}(1-\epsilon)}\\
	&\geq \frac{(1-\epsilon)\cdot {\kappa_d}\cdot R^2\cdot 2^{-\left \lceil \log_2(d+1)\right \rceil}}{(1-\epsilon)R^2\cdot 2^{-\left \lceil \log_2(d+1)\right \rceil} + 1-R+2^{-\tau}(1-\epsilon)}
\end{align*}
where the first inequality holds since $\frac{K_m}{N_m}\xrightarrow{m\to \infty}R$ and $\alpha_m,\beta_m\xrightarrow{m\to \infty}0$, and the last inequality holds since $R\in (0,1)$. Finally, taking $\epsilon\downarrow 0$, we obtain the rate lower bound in the statement of the Lemma.

We now prove that the rate lower bound derived above, is achievable over a $(d,\infty)$-RLL input-constrained BMS channel, using the decoding procedure given in (\textbf{D1})--(\textbf{D2}), {as} long as $R<C$, {where $C$ is the capacity of the unconstrained channel}. To this end, note that Step (\textbf{D1}) decodes the $L$ parts, $\mathbf{c}_1,\ldots,\mathbf{c}_L$, each with probability of error at most $\eta_m>0$, with $\eta_m\xrightarrow{m\to \infty} 0$, if $R<C$. Hence, the {block} probability of error Pr$\left[\hat{c}_{K_m}^{N_m-1}\neq {c}_{K_m}^{N_m-1}\right]$ of the decoding stage (\textbf{D1}), is at most $L\cdot \eta_m$. Moreover, given the event $\left\{\hat{c}_{K_m}^{N_m-1} = {c}_{K_m}^{N_m-1}\right\}$, Step (\textbf{D2}) determines the information bits $\mathbf{w}$, with {conditional block} probability of error Pr$\left[\hat{c}_0^{K_m-1}\neq \mathbf{w} \big\vert\ \hat{c}_{K_m}^{N_m-1} = {c}_{K_m}^{N_m-1}\right] \leq \delta_m$, with $\delta_m>0$ such that $\delta_m\xrightarrow{m\to \infty} 0$, if $R<C$. Hence, the overall  probability of correct estimation of the information bits $\mathbf{w}$ can be bounded as
\begin{align*}
\text{Pr}\left[\hat{c}_0^{K_m-1}= \mathbf{w}\right]&\geq \text{Pr}\left[\hat{c}_{K_m}^{N_m-1}= {c}_{K_m}^{N_m-1}\right]\cdot \text{Pr}\left[\hat{c}_0^{K_m-1}= \mathbf{w} \ \big\vert \ \hat{c}_{K_m}^{N_m-1} = {c}_{K_m}^{N_m-1}\right]\\
&\geq (1-L\cdot \eta_m)\cdot (1-\delta_m).
\end{align*}
As $L$ is a constant, the lower bound on the probability of correct estimation converges to $1$, as $m\to \infty$.
\end{proof}

In Section \ref{sec:main}, we compared the achievable rate of $\frac{C}{2}$, for $d=1$, obtained using linear subcodes of RM codes in Theorem \ref{thm:rm} ({using the block-MAP decoder of the larger RM codes}), with rates achievable by the two-stage coding scheme above. An equivalent way of stating our observations there is that the two-stage coding scheme achieves a higher rate for erasure probabilities $\epsilon\lessapprox 0.2387$, when used over the $(1,\infty)$-RLL input-constrained BEC, and for crossover probabilities $p$ that are in the approximate interval $(0,0.0392)\cup(0.9608,1)$, when used over the $(1,\infty)$-RLL input-constrained BSC.


\section{Conclusion}
\label{sec:conclusion}
This work proposed explicit, deterministic coding schemes, without feedback, for binary memoryless symmetric (BMS) channels with $(d,\infty)$-runlength limited (RLL) constrained inputs. In particular, achievable rates were calculated by identifying specific constrained linear subcodes and proving the existence of constrained, potentially non-linear subcodes, of a sequence of RM codes of rate $R$. Furthermore, upper bounds were derived on the rate of the largest linear $(d,\infty)$-RLL subcodes of RM codes of rate $R$, showing that the linear subcodes identified in the constructions were essentially rate-optimal. A novel upper bound on the rate of general $(1,\infty)$-RLL subcodes was derived using properties of the weight distribution of RM codes, and a new explicit two-stage coding scheme was proposed, whose rates are better than coding schemes using linear subcodes, for large $R$.




There is much scope for future work in this line of investigation. Firstly, following the close relationship between the sizes of $(1,\infty)$-RLL subcodes and the weight distributions of RM codes established in this work, a more sophisticated analysis of achievable rates can be performed with the availability of better lower bounds on the weight distributions of RM codes. Likewise, sharper upper bounds on the weight distributions of RM codes will also lead to better upper bounds on the rate of any $(1,\infty)$-RLL subcodes of a certain canonical sequence of RM codes. It would also be of interest to derive good upper bounds on the rates of general $(d,\infty)$-RLL subcodes of RM codes, via weight distributions or otherwise.

Further, we wish to explore the construction of general doubly-transitive linear codes that contain a large number of $(d,\infty)$-RLL constrained codewords. Such constructions will also help lend insight into the capacity of input-constrained BMS channels without feedback, which is a well-known open problem.

\appendices
\section{Proof of Inequality \eqref{eq:inter}}
In this section, we show that the inequality $\beta(m)\leq \theta(m)$ holds, for large $m$ and sufficiently small $\delta>0$, with $\beta(m)$ and $\theta(m)$ defined in equations \eqref{eq:beta} and \eqref{eq:inter}, respectively. 

We start with the expression 
\begin{equation}
	\beta(m) = 2^{o(2^m)} \cdot \sum_{i=1}^{r_m-1} \exp_2 \left(2^{m-1}\cdot {R_{m,+}} \cdot h_b(2^{-i}) - {m-2-i \choose \leq r_m}\right)
	\label{eq:beta:2}
\end{equation}
We will split the sum $\sum\limits_{i=1}^{r_m-1}$ into two parts: $\sum\limits_{i=1}^{t_m}$ and $\sum\limits_{i={t_m}+1}^{r_m-1}$, where $t_m := \lfloor m^{1/3} \rfloor$.  For $i \in [t_m+1:r_m-1]$, we have
$$
2^{m-1}\cdot {R_{m,+}} \cdot h_b(2^{-i}) - {m-2-i \choose \leq r_m} \ \leq \ 2^{m-1}\cdot {R_{m,+}} \cdot h_b(2^{-i})  \ = \ o(2^m),
$$
since $h_b(2^{-i})\leq h_b(2^{-t_m-1})$, and $h_b(2^{-t_m-1})\xrightarrow{m\to \infty} 0$, by the continuity of entropy. Thus, the contribution of each term of the sum $\sum\limits_{i={t_m}+1}^{r_m-1}$ is $2^{o(2^m)}$, and since there are at most $r_m = O(m)$ terms in the sum, the total contribution from the sum is $2^{o(2^m)}$. 

Turning our attention to $i \in [t_m]$, we note first that for $m$ large enough, we have ${R_{m,+}}\leq R(1+\epsilon)$, for $\epsilon \in (0,1)$ suitably small. Also, we write
\begin{align}
2^{m-1}\cdot {R_{m,+}} \cdot & h_b(2^{-i}) - {m-2-i \choose \leq r_m} \notag \\ & \ \ \ \ \leq \ 2^{m-1} \cdot \left[ R(1+\epsilon)\cdot h_b(2^{-i}) - 2^{-(i+1)}\cdot \frac{1}{2^{m-2-i}} {m-2-i \choose \leq r_m}\right]. \label{i_in_tm_ineq}
\end{align}
By Lemma~\ref{lem:rate}, we obtain that $\frac{1}{2^{m-2-i}} {m-2-i \choose \leq r_m}$ converges to $R$ for all $i \in [t_m]$. In fact, with a bit more effort, we can show that this convergence is uniform in $i$. 
Indeed, since $i \le t_m$, by virtue of \eqref{eq:rm_diff}, we have $|r_m - r_{m-2-i}|  \le  \frac{t_m+2}{2} + \frac{\sqrt{t_m+2}}{2} \lvert Q^{-1}(1-R)\rvert + 1 =: \nu_m$. Using the notation in the proof of Lemma~\ref{lem:rate}, we have $\frac{1}{2^{m-2-i}} {m-2-i \choose \leq r_m} = \Pr[S_{m-2-i} \le r_m]$. Thus, analogous to \eqref{sandwich}, we have for all sufficiently large $m$,
\begin{align*}
& \Pr[\overline{S}_{m-2-i} \le Q^{-1}(1-R) - \frac{\nu_m}{\frac12 \sqrt{m-2-t_m}}]  \le \ \Pr[S_{m-2-i} \le r_m] \le \ \Pr[\overline{S}_{m-2-i} \le Q^{-1}(1-R) + \frac{\nu_m}{\frac12 \sqrt{m-2-t_m}}].
\end{align*}
Now, we apply the Berry-Esseen theorem (see e.g., \cite[Theorem~3.4.17]{durrett}) which, in this case, asserts that \newline $\left\lvert\Pr[\overline{S}_m \le x] - \Pr[Z \le x]\right\rvert \le 3/\sqrt{m}$, for all $x \in \mathbb{R}$ and positive integers $m$, where $Z \sim N(0,1)$. Thus, \newline $\left\lvert\Pr[\overline{S}_{m-2-i} \le x] - \Pr[Z \le x]\right\rvert \le \frac{3}{\sqrt{m-2-i}} \le \frac{3}{\sqrt{m-2-t_m}}$ holds for all $x \in \mathbb{R}$ and $i \in [t_m]$. This yields
\begin{align*}
& \Pr[Z \le Q^{-1}(1-R) - \frac{\nu_m}{\frac12 \sqrt{m-2-t_m}}] - \frac{3}{\sqrt{m-2-t_m}} \\
& \ \ \le \ \Pr[S_{m-2-i} \le r_m]  \\
& \ \ \ \ \le \ \Pr[Z \le Q^{-1}(1-R) + \frac{\nu_m}{\frac12 \sqrt{m-2-t_m}}] + \frac{3}{\sqrt{m-2-t_m}}.
\end{align*}
Since $t_m$ and $\nu_m$ are both $o(\!\!\sqrt{m})$, we deduce that, as $m \to \infty$, $\Pr[S_{m-2-i} \le r_m] = \frac{1}{2^{m-2-i}} {m-2-i \choose \leq r_m}$ converges to $R$ uniformly in $i \in [t_m]$. 

Hence, for small $\epsilon \in (0,1)$ and $m$ large enough, we have that for all $i \in [t_m]$ that
\[
\frac{1}{2^{m-2-i}} {m-2-i \choose \leq r_m} \geq (1 - \epsilon) R,
\]
so that, carrying on from \eqref{i_in_tm_ineq}, we have that
\begin{align}
R(1+\epsilon)\cdot h_b(2^{-i})-2^{-(i+1)}\cdot \frac{1}{2^{m-2-i}}{m-2-i\choose\leq r_m} \leq R\left[(1+\epsilon)\cdot h_b(2^{-i})-2^{-(i+1)}\cdot (1-\epsilon)\right]. \label{tm_ineq}
\end{align}
Now, we claim that for any $i\in \mathbb{N}$, the expression within square brackets above can be bounded above as:
\begin{align}
	(1+\epsilon)\cdot h_b(2^{-i})-2^{-(i+1)}\cdot (1-\epsilon)&\leq 
\frac{3}{4}+2\epsilon,
	\label{eq:interineq}
\end{align}
The proof of equation \eqref{eq:interineq} relies on simple calculus, and is relegated to Appendix B.

To put it all together, recall that we split the sum $\sum\limits_{i=1}^{r_m-1}$ in \eqref{eq:beta:2} into two parts: $\sum\limits_{i=1}^{t_m}$ and $\sum\limits_{i={t_m}+1}^{r_m-1}$, where $t_m := \lfloor m^{1/3} \rfloor$. Given an arbitrarily small $\delta>0$, settng $\epsilon = \delta/2$, we obtain via \eqref{i_in_tm_ineq}--\eqref{eq:interineq} that, for all sufficiently large $m$, the contribution from the sum $\sum_{i=1}^{t_m}$ is at most
$$
m^{1/3} \cdot \exp_2\left(2^{m-1}R\cdot\left(\frac34 +\delta\right) \right).
$$
We noted earlier that the contribution from the sum $\sum_{i=t_m+1}^{r_{m}-1}$ is $\text{exp}_2(o(2^m))$. Therefore, the overall sum $\sum\limits_{i=1}^{r_m-1}$ in \eqref{eq:beta:2} can be bounded above, for all sufficiently large $m$, by 
$$
2m^{1/3} \cdot \exp_2\left(2^{m-1}R\cdot\left(\frac34 +\delta\right) \right),
$$
Consequently,
$$
\beta(m) \ \le \ \text{exp}_2\left(2^{m-1}R\cdot\left(\frac34 +\delta\right)+o(2^m)\right) \ =: \ \theta(m).
$$
\section{Proof of inequality \eqref{eq:interineq}}
In this section, we show that the inequality in \eqref{eq:interineq} holds. 

Firstly, we note that the expression on the left of \eqref{eq:interineq} obeys, for $i\geq 1$,
\begin{align}
	(1+\epsilon)\cdot h_b(2^{-i})-2^{-(i+1)}\cdot (1-\epsilon)&= h_b(2^{-i})-2^{-(i+1)}+\epsilon\left(h_b(2^{-i})+2^{-(i+1)}\right) \notag\\
	&\leq h_b(2^{-i})-2^{-(i+1)}+2\epsilon. \label{eq:appb1}
\end{align}

Now, consider the function $f(x) = h_b(2^{-x})-2^{-(x+1)}$, where $x\in [0,\infty)$. By taking the derivative with respect to $x$ on both sides, we get
\[
f^\prime(x)=(2^{-x}\ln 2) \cdot \left[\frac{1}{2}-\log_2\left(\frac{1-2^{-x}}{2^{-x}}\right)\right].
\]
The term within square brackets above is positive when $x\in \left(0,\log_2(1+\sqrt{2})\right)$, and is negative for $x\in \left(\log_2(1+\sqrt{2}),\infty\right)$. Importantly, this implies that $f$ is decreasing in the interval $[2,\infty)$. Furthermore, we note that $f(1) = h_b(\frac12)-\frac14 = \frac34$, and $f(2) = h_b(\frac14)-\frac18 \approx 0.68 < f(1)$. Hence, $f$ is decreasing over integers $i\in \mathbb{N}$.

With this, we obtain that the right side of the inequality in \eqref{eq:appb1} is at most $\frac34+2\epsilon$.


\ifCLASSOPTIONcaptionsoff
  \newpage
\fi



%
\bibliographystyle{IEEEtran}
{\footnotesize
	\bibliography{references}}

\end{document}